\documentclass[twoside,11pt,final]{entics} 

\usepackage{amsfonts,amsmath,amssymb}
\usepackage{stmaryrd}
\usepackage{proof}

\usepackage{enticsmacro}

\usepackage{FMC}

\renewcommand\e{\varepsilon}
\renewcommand\Z{\mathbb{I}}

\newcommand\msep{\kern2pt}

\newcommand\subs{s}
\newcommand\wh{\mathsf{wh}}

\newcommand\vc[1]{\vcenter{\hbox{\ensuremath{#1}}}}

\newcommand\coloneqq{\mathrel{{:}{:}{=}}}

\newcommand\fv{\mathsf{fv}}

\newcommand\cbv{\textsf{cbv}}
\newcommand\cbn{\textsf{cbn}}

\newcommand\cbpv{\textsf{cbpv}}

\newcommand\val{\mathsf{v}}

\newcommand\dup{\mathsf{d}}
\newcommand\aff{\mathsf{a}}

\newcommand\ev[2][30]{~\tikz[baseline] \draw[|->] (0,3pt)--node[above]{$\scriptstyle\term{#2}$}(#1pt,3pt) ; ~}
\newcommand\evs[2][30]{~\tikz[baseline] \draw[|->>] (0,3pt)--node[above]{$\scriptstyle\term{#2}$}(#1pt,3pt) ; ~}

\newcommand\SEQ[1]{\textsc{run}(\term{#1})}
\newcommand\RUN[1]{\textsc{run}(\type{#1})}

\newcommand\EVS[1]{\textsc{eval}(\term{#1})}
\newcommand\EVL[1]{\textsc{eval}(\type{#1})}

\newcommand\src{\mathsf{src}}
\newcommand\tgt{\mathsf{tgt}}

\newcommand\If{\mathsf{if}}
\newcommand\Then{\mathsf{then}}
\newcommand\Else{\mathsf{else}}
\newcommand\While{\mathsf{while}}
\newcommand\Break{\mathsf{break}}
\newcommand\Return{\mathsf{return}}
\newcommand\Throw{\mathsf{throw}}
\newcommand\Try{\mathsf{try}}
\newcommand\Catch{\mathsf{catch}}
\newcommand\Case{\mathsf{case}}
\newcommand\Of{\mathsf{of}}
\newcommand\Do{\mathsf{do}}

\newcommand\bind{\mathbin{
\begin{tikzpicture}[x=1pt,y=-1pt]
\draw[cap=round,join=round,line width=0.7pt] (0,2)--(2,0)--(0,-2) (2,2)--(4,0)--(2,-2) (4.5,1)--(7,1) (4.5,-1)--(7,-1);
\end{tikzpicture}}}

\makeatletter
\newcommand\head{\@ifstar{\headrotate}{\headfixed}}
\newcommand\headrotate[1]{\rotatebox{90}{#1}}
\newcommand\headfixed[1]{#1\rule[-2pt]{0pt}{13pt}}
\makeatother

\newcommand\sr[1]{\!\!\scriptstyle{\SR{#1}}}
\newcommand\SR[1]{\mathsf{%
  \ifx#1x var\else%
  \ifx#1a push\else%
  \ifx#1l pop\else%
  \ifx#1c chc\else%
  \ifx#1; case\else%
  \ifx#1o loop\else%
  \ifx#1e exp\else%
  \ifx#1i incl\else%
  \ifx#1s stk0\else%
  \ifx#1S stk1\else%
  \ifx#1M mem\else%
  \ifx#1k cnt0\else%
  \ifx#1K cnt1\else%
  \ifx#1+ state\else%
  \fi\fi\fi\fi%
  \fi\fi\fi\fi%
  \fi\fi\fi\fi\fi\fi%
}}

\def\conf{MFPS 2025}
\volume{5}			    
\def\volu{5	}		
\def\papno{9}			
\def\mydoi{16682}

\def\lastname{Heijltjes}
\def\runtitle{FMC III: Control} 
\def\copynames{W. Heijltjes}

\begin{document}

\begin{frontmatter}

\title{The Functional Machine Calculus III: Control}

\author{Willem Heijltjes}

\address{%
Department of Computer Science\\
University of Bath\\
United Kingdom}
						
\thanks{Email: \href{mailto:w.b.heijltjes@bath.ac.uk} {\texttt{\normalshape w.b.heijltjes@bath.ac.uk}}}

\thanks{Thanks to Ohad Kammar, Guy McCusker, Olivier Laurent, and Nicolas Wu for discussions and feedback.}

\maketitle

\begin{abstract}
The Functional Machine Calculus (Heijltjes 2022) is a new approach to unifying the imperative and functional programming paradigms. It extends the lambda-calculus, preserving the key features of confluent reduction and typed termination, to embed computational effects, evaluation strategies, and control flow operations. The first instalment modelled sequential higher-order computation with global store, input/output, probabilities, and non-determinism, and embedded both the call--by--name and call--by--value lambda-calculus, as well as Moggi's computational metalanguage and Levy's call--by--push--value. The present paper extends the calculus from sequential to branching and looping control flow. This allows the faithful embedding of a minimal but complete imperative language, including conditionals, exception handling, and iteration, as well as constants and algebraic data types.

The calculus is defined through a simple operational semantics, extending the (simplified) Krivine machine for the lambda-calculus with multiple operand stacks to model effects and a continuation stack to model sequential, branching, and looping computation. It features a confluent reduction relation and a system of simple types that guarantees termination of the machine and strong normalization of reduction (in the absence of iteration). These properties carry over to the embedded imperative language, providing a unified functional--imperative model of computation that supports simple types, a direct and intuitive operational semantics, and a confluent reduction semantics.
\end{abstract}

\begin{keyword}
lambda-calculus, computational effects, exception handling, simple types
\end{keyword}

\end{frontmatter}

\section{Introduction}

An interesting challenge in programming language theory is to find good ways of extending the $\lambda$-calculus with computational effects. The $\lambda$-calculus, established by Landin's seminal work~\cite{Landin-1964,Landin-1966} as the basis of the \emph{functional} programming paradigm, gives a canonical treatment of higher-order functions that sets the standard for good semantic properties such as compositionality, referential transparency, and type safety. However, these do not readily extend to effects. The \emph{imperative} paradigm, meanwhile, sets expectations for a treatment of effects: a clear syntax with a direct, intuitive operational meaning and the seamless integration of multiple effects. 

The overarching aim of the challenge is thus to reconcile both paradigms. This requires a unified model of computation that naturally accommodates both higher-order functions and sequential operations, while supporting types, a confluent reduction semantics, and other means of reasoning. In line with the ambition of these requirements, approaches to address this challenge include some of the most influential developments in the field, such as Moggi's account of effects as \emph{monads}~\cite{Moggi-1991}, Levy's \emph{call--by--push--value} (\cbpv)~\cite{Levy-2003}, and Plotkin and Pretnar's \emph{effect handlers}~\cite{Plotkin-Pretnar-2013}.

The Functional Machine Calculus (FMC) is a new approach to this challenge. The idea is to take the Krivine Machine~\cite{Krivine-2007} seriously as the notion of \emph{sequentiality} in the $\lambda$-calculus. This provides a direct operational semantics in a simple abstract machine with a single stack, where \emph{application} is \emph{push}, \emph{abstraction} is \emph{pop}, and \emph{variable} is \emph{execute}. The design principle is then to extend the machine in natural and minimal ways to capture a wide range of imperative features, maintaining the calculus simultaneously as a direct instruction language for the machine, and as a calculus that supports confluent reduction and simple types: a \emph{``machine calculus''}, a machine language that is also a calculus. Previously, the following two extensions were introduced~\cite{Barrett-Heijltjes-McCusker-2023,Heijltjes-2022}.

\begin{description}
	\item[Sequencing]
The $\lambda$-calculus is extended with imperative \emph{sequencing} and \emph{skip}, its unit, as \emph{composition} and \emph{identity} for stack operations. This gives a model of higher-order sequential computation that encompasses both call--by--name and call--by--value behaviour. It faithfully embeds Plotkin's call--by--value (\cbv) $\lambda$-calculus~\cite{Plotkin-1975}, Moggi's computational metalanguage~\cite{Moggi-1991}, and Levy's \cbpv~\cite{Levy-2003}. Sequencing is implemented on the machine in standard fashion with a continuation stack (elided in the original presentation but introduced in~\cite{Barrett-Castle-Heijltjes-2024,Heijltjes-Majury-2025}).

	\item[Locations]
The machine is generalized from one to many operand stacks, each named by a \emph{location} and directly accessed through push- and pop-instructions. These may then model various computational effects: \emph{mutable higher-order store} as stacks of depth one; \emph{input/output} as pop-only and push-only streams; and \emph{probabilities} and \emph{non-determinism} as probabilistically, respectively non-deterministically generated streams.
\end{description}

\noindent
This paper presents the following third extension.

\begin{description}
	\item[Control]
The machine is generalised from strictly sequential to \emph{branching} and \emph{looping} control flow. The \emph{skip} command is replaced with a set of \emph{choice} labels, each indicating a branch of the computation, while sequencing becomes conditional on a choice, composing on the given branch only. A \emph{loop} construct is introduced that repeats on a given choice and exits otherwise. Computationally, choice labels represent \emph{constants}, \emph{exceptions}, and \emph{data constructors}. Various notions of control flow then embed into the FMC: \emph{exception handling}, \emph{conditionals} and more generally \emph{algebraic data types}, and \emph{iteration} with \emph{escape}.
\end{description}

\noindent
What these control constructs have in common is that, semantically, they are modelled by \emph{sums} or \emph{coproducts}. The exception monad $TX = E + X$ is the coproduct of the value type $X$ with a type $E$ for exceptions. The data type of booleans is the sum $1+1$ and a conditional is a co-diagonal $[f,g]:1+1\to X$, which algebraic data types generalise to sums of products indexed by named constructors. Iteration is modelled by taking a morphism $f:A\to A+B$ to one $\mathsf{iter}\,f:A\to B$, which loops on $A$ and exits on $B$~\cite{Bloom-Esik-1993} (the construction is semantically dual to recursion~\cite{Simpson-Plotkin-2000}).

The purpose of this paper is to capture these constructs in the FMC with a minimal and natural extension to the machine and the calculus. The concrete contributions are the calculus itself, its small-step operational semantics in the extended stack machine, a big-step operational semantics, a confluent reduction relation, and a notion of simple types that (in the absence of loops) guarantees successful termination of the machine and strong normalization of reduction. The proofs of these properties moreover employ standard techniques, in straightforward but, in some cases, novel ways. Confluence (Section~\ref{sec:confluence}) is by \emph{parallel reduction}~\cite{Takahashi-1995}. Machine termination (Section~\ref{sec:termination}) is by a reducibility argument~\cite{Tait-1967} with a direct induction on typing derivations, without the combinatorial reasoning typical of strong normalization proofs. Strong normalization (Section~\ref{sec:normalization}) is by extending big-step evaluation to count transitions on the machine, and then demonstrating that reduction shortens machine evaluation, thus separating the logical content (again by reducibility) from the combinatorics (comparing reduction against the machine).

The FMC itself is a minimal language of six syntactic constructs, each representing a natural instruction on a simple machine with multiple operand stacks and one continuation stack. It nevertheless faithfully embeds a complete imperative language with store, exception handling, and loops, as well as the \cbv\ $\lambda$-calculus~\cite{Plotkin-1975}, the computational metalanguage~\cite{Moggi-1991}, and \cbpv~\cite{Levy-2003}. Moreover, these embeddings are by \emph{macro-expansion}, where the encoded constructs are given as definitions (or \emph{syntactic sugar}) over FMC terms. As a consequence, confluent reduction and simple types are directly conferred on imperative constructs. This suggests that, in Levy's terminology~\cite{Levy-1999}, these embeddings may be conjectured to be \emph{subsuming}: the semantics of the encoded constructs is expected to arise as that of their FMC encoding. This is immediate for the operational semantics, as presented in this paper. It was recently confirmed for \emph{intersection types}: established notions of intersection typing for store~\cite{Davies-Pfenning-2000,deLiguoro-Treglia-2021} arise naturally via the embedding of store in the intersection-typed FMC with \emph{sequencing} and \emph{locations}~\cite{Heijltjes-2025}. Other notions of semantics will be investigated in future.

More specifically, the FMC is a new solution to the problem of combining multiple effects. Both user-defined monads such as \emph{state} $TX=S\to (S\times X)$ and user-defined effect handlers support confluence and types, by encoding effects into the $\lambda$-calculus with inductive data types. But monads do not compose, and require monad transformers to combine multiple effects~\cite{Liang-Hudak-Jones-1995} (their limitations are discussed in \cite{Kiselyov-Sabry-Swords-2013}). Effect handlers create a two-layered system, separating the \emph{handlers} from the \emph{effect operations}, in particular complicating the treatment of exceptions since \emph{raising} is an effect operation while \emph{catching} is a handler. As a consequence, translating imperative constructs into monads or handlers is more involved. The main contribution of the FMC is to allow a direct embedding of a minimal but complete imperative language, seamlessly combining multiple effects within a minimal, typed functional calculus.

While monads are universal for effects, and handlers for algebraic effects~\cite{Plotkin-Power-2002}, the FMC as presented here covers a substantial but fixed set of effects only. This is being addressed in a present line of enquiry by introducing effect handlers into the FMC, which appears a very natural combination. Independent of the present \emph{control} extension, the sequential FMC has been extended with \emph{non-deterministic} branching for relational computation~\cite{Barrett-Castle-Heijltjes-2024}, and \emph{probabilistic} branching to capture probabilistic choice in the $\lambda$-calculus~\cite{Heijltjes-Majury-2025}. This generalises to an \emph{``algebraically''} branching FMC that fits neatly into the effect handlers paradigm: branching is given by abstract effect operators, which are interpreted computationally by handlers~\cite{Plotkin-Pretnar-2013}. Interesting questions then arise regarding the relation between the different notions of branching, the ``algebraic'' approach versus the present \emph{control} extension, and how each encodes effects. These will be considered in future work.

The next three sections will introduce the terms and types of the FMC informally and explore the connections with related work, while the remaining sections will present the formal development. Starting from the $\lambda$-calculus, the three extensions \emph{sequencing} (Section~\ref{sec:sequencing}), \emph{locations} (Section~\ref{sec:locations}), and \emph{control} (Section~\ref{sec:control}) will be covered in turn, while building up the interpretation of an imperative language complete with store, conditionals, exceptions, and loops. The calculus will be illustrated and motivated through an informal presentation of its operational semantics; the abstract machine itself is given in Section~\ref{sec:FMC} at the start of the formal development. The formulation of \emph{sequencing} and \emph{locations} has matured since the original presentation~\cite{Heijltjes-2022}, mainly due to the requirements of the present \emph{control} extension. A preliminary version was presented at MFPS 2024~\cite{Heijltjes-2024}.


\section{Sequencing}
\label{sec:sequencing}

The Krivine Abstract Machine (KAM)~\cite{Krivine-2007}, here in a form that uses substitution rather than environments, gives a natural perspective of $\lambda$-terms as stack operations, where application \emph{pushes} and abstraction \emph{pops}. To emphasize this view and to accommodate the extensions of the FMC, application $M\,N$ is written $\term{[N].M}$ (``push $\term N$'', c.f.\ Levy's $V`M$~\cite{Levy-2003}) and abstraction $\lambda x.M$ is written $\term{<x>.M}$ (``\emph{pop} and bind to $\term x$''). The operational semantics of a term $\term M$ acting on a stack $S$ will be depicted as follows. 
\[
	S\evs{M}
\]
Writing a stack with the head to the right, as $S\,\term M$, below left is the evaluation of $\term{[N].M}$, which pushes $\term N$ to the stack $S$ and continues by evaluating $\term M$. On the right, $\term{<x>.M}$ pops the first term $\term N$ off the stack and substitutes it for $\term x$ in the remaining computation, as $\term{\{N/x\}M}$.
\[
	\term{[N].M}:~~S \ev{[N]} S\,\term N \evs{M} \qquad\qquad
	\term{<x>.M}:~~S\,\term N \ev{<x>} S \evs{\{N/x\}M}
\]
Note that the result of the computation is left implicit. This is because the KAM returns the first term without a transition, which is either a (free) variable or an abstraction with an empty stack, while the FMC takes a different approach: FMC terms return the stack. This is effected by the \emph{sequencing} extension, which introduces imperative \emph{skip} $(\term{*})$ and \emph{sequencing} $(\term{M;N})$ as the \emph{identity} and \emph{composition} of terms as functions on stacks (or rather \emph{partial} functions, because of non-termination).
\[
	\term{*}:~~S \evs{*} S \qquad\qquad \term{M;N}:~~ R \evs{M} S \evs{N} T
\]
While natural from the perspective of the machine, this notion of sequentiality represents a significant conceptual departure from the traditional $\lambda$-calculus. \emph{Skip} indicates \emph{successful termination}, as it does in imperative models, and returns the stack, enabling the view of terms as partial functions on stacks. Sequential composition ($\term{M;N}$) is then the natural notion of composition, and likewise has its standard imperative meaning, first evaluate $\term M$ then evaluate $\term N$. Both represent a departure from traditional notions of termination and composition in $\lambda$-calculus, but are commensurate with sequentiality in \cbpv\ and the computational metalanguage, as will be explored below.

Introducing both constructs into the $\lambda$-calculus gives the \emph{sequential $\lambda$-calculus}, a fragment of the FMC.
\[
	\term{M,N} ~\coloneqq~ \term x ~\mid~ \term{<x>.M} ~\mid~ \term{[N].M} ~\mid~ \term{*} ~\mid~ \term{M;N}
\]
The extension gives terms for natural stack operations: e.g.\ $\term{<x>.[x].[x].*}$ duplicates the head of the stack, $\term{<x>.*}$ deletes it, $\term{<x>.<y>.[x].[y].*}$ exchanges the top two elements, and $\term{<x>.x}$ (the traditional identity term) pops and executes the first stack element. Note that since stacks are last--in--first--out, terms of the form $\term{<x_1>..<x_n>.[x_n]..[x_1].*}$, with the order of variables reversed between popping and pushing, return an input stack (of depth at least $n$) unchanged. By way of example, below is the evaluation of $\term{<f>.<g>.(f;g)}$ which pops and executes the first two stack elements. 
\[
	\term{<f>.<g>.(f;g)}:
	~~
	R\,\term{N}\,\term{M}  \ev{<f>} 
	R\,\term{N}            \ev{<g>}
	R                      \evs{M}
	S                      \evs{N}
	T
\]
The present formulation of the calculus differs from the original~\cite{Heijltjes-2022,Barrett-Heijltjes-McCusker-2023}. This had the variable construct as a prefix, $\term{x.M}$, and composition $\term{M;N}$ as a defined operation, promoting a view of terms as sequences of pushes, pops, and variables. Taking \emph{skip} and \emph{sequence} as primitives is now preferred, for its clear and direct combination of $\lambda$-calculus with sequential computation, and as a prerequisite for the \emph{control} extension of this paper. Both formulations are essentially interchangeable: the previous sequential variable $\term{x.M}$ is a restriction of sequential composition as $\term{x;M}$, and the previous sequencing operation is implemented here through the reduction relation (see Section~\ref{sec:FMC}).

Higher-order stack calculi similar to the sequential $\lambda$-calculus have appeared before. Notable are Hasegawa's $\kappa$-calculus~\cite{Hasegawa-1995}, revisited by Power and Thielecke~\cite{Power-Thielecke-1999}, a study of compiler languages by Douence and Fradet~\cite{Douence-Fradet-1998}, and the \emph{concatenative} programming paradigm~\cite{Herzberg-Reichert-2009}, meaning higher-order stack languages such as $\lambda$-Forth, Joy, and Factor~\cite{Pestov-Ehrenberg-Groff-2010}. Typing, explored next, is familiar from several of the above calculi and languages, and from other stack languages like WebAssembly~\cite{Rossberg-etal-2018}.


\subsection{Stack typing}

Types for the sequential $\lambda$-calculus follow the view of terms as functions on stacks. A type is an implication $\type{!s=>!t}$ between type vectors $\type{!s}$ and $\type{!t}$, which represent the types of the input and output stacks, as below. The empty type vector is $\type\e$. The antecedent vector $\type{!s}$ of a type $\type{!s=>!t}$ is implicitly reversed, so that identity types $\type{!s=>!s}$ follow the shape of identity terms, $\type{s_1..s_n=>s_n..s_1}$.
\[
	\type{s,t} ~\coloneqq~ \type{!s=>!t} \qquad \type{!s,!t}~\coloneqq~\type{t_1..t_n}
\]
The intuitive meaning of a type assignment $\term{M:!s=>!t}$ is then as follows: for any stack $S$ typed by $\type{!s}$ the term $\term M$ will evaluate successfully and return a stack $T$ typed by $\type{!t}$. This is formalized as a reducibility predicate in Section~\ref{sec:termination} to prove termination of the machine.
\[
	\term{M:!s=>!t} \quad\Longrightarrow\quad \forall S:\type{!s}.~\exists T:\type{!t}.~S \evs{M} T
\]
Typical of stack typing is further the notion of \emph{stack expansion}: a term $\term{M:?s=>!t}$ may also be typed $\term{M:?s\,?r=>!r\,!t}$, since if $\term M$ evaluates with a stack $S$, it also evaluates with a larger stack $R\,S$ (using juxtaposition for concatenation), leaving $R$ untouched.
\[
	S \evs{M} T  \quad\Longrightarrow\quad  R\,S \evs{M} R\,T
\qquad\qquad
	\term{M:?s=>!t} \quad\Longrightarrow\quad \term{M:?s\,?r=>!r\,!t}
\]

Semantically, stacks are products of terms, and the calculus forms a strict Cartesian closed category~\cite{Barrett-Heijltjes-McCusker-2023}, where morphisms are closed terms, identity is $\term{*:?t=>!t}$, and composition is $\term{M;N:?r=>!t}$ for $\term{M:?r=>!s}$ and $\term{N:?s=>!t}$. Other example morphisms are a projection $\term{<x>.<y>.[x].* : st=>s}$, a diagonal $\term{<x>.[x].[x].* : t=>tt}$, and a terminal map $\term{<x>.* : t=>\e}$. Interestingly, in this interpretation the traditional identity term $\lambda x.x$ becomes \emph{eta}/\emph{eval} $\term{<x>.x: (!s=>!t)\,!s=>!t}$. 

Typing is conservative over simple types for the $\lambda$-calculus. These embed as \emph{input-only} types $\type{!t=>\e}$, with empty output vector, in accordance with the observation that \emph{skip} is needed to yield a return stack. The base type $o$ is interpreted as $\type{\e=>\e}$, and the arrow type $\sigma\smallbin\to\tau$ adds a further input type $\sigma$ to the antecedent vector of $\tau$, as given below. The overall picture is given by $\typ{t_1}\smallbin\to\dots\smallbin\to\typ{t_n}\smallbin\to o=\type{t_1..t_n=>\e}$.
\[
	o~=~\type{\e=>\e} \qquad \type{s}\smallbin\to(\type{!t=>\e})~=~\type{s\,!t=>\e}
\]
Two further observations will be made that set the sequential $\lambda$-calculus apart from the $\lambda$-calculus. Firstly, all types are inhabited (see~\cite[Remark 3.7]{Heijltjes-2022} and Proposition~\ref{prop:inhabitation}). In particular, the embedded base type $o=\type{\e=>\e}$ is inhabited by $\term*$. Secondly, the calculus has a natural first-order restriction, an internal language for Cartesian categories, as follows (see also~\cite{Hasegawa-1995}). This is expanded to a model of \emph{relational} computation in~\cite{Barrett-Castle-Heijltjes-2024}.
\[
	\term{M,N} ~\coloneqq~ \term{<x>.M} ~\mid~ \term{[x].M} ~\mid~ \term* ~\mid~ \term{M;N}
\]


\subsection{Expressing evaluation strategies}

Along the same lines as previous higher-order stack calculi~\cite{Douence-Fradet-1998,Power-Thielecke-1999}, the sequential $\lambda$-calculus embeds Plotkin's \cbv\ $\lambda$-calculus~\cite{Plotkin-1975} and Moggi's computational metalanguage~\cite{Moggi-1991}. The key idea in both interpretations is that \emph{values}, respectively \emph{return values}, are returned on the stack. The former then embeds as below left, using subscript $-_\val$ to distinguish it from the call--by--name calculus, with $@_\val$ for application. This is a typed embedding, with types for values embedded as $\typ{s}\smallbin{\to_\val}\typ{t}=\type{s=>t}$ and for computations as $\typ{t}_\val=\type{\e=>t}$. The latter embeds as below right, likewise passing a single return value on the stack, and again this is typed, with $\typ{s_1}\smallbin\to\dots\smallbin\to\typ{s_n}\smallbin\to T(\typ t)$ for a monad $T$ embedding as $\type{s_1..s_n=>t}$.
\[
\begin{aligned}
	         x_\val &~=~\term{[x].*}			&&& \mathsf{return}\,M                &~=~ \term{[M].*}
\\ \lambda_\val x.M &~=~\term{[<x>.M].*}        &&& \mathsf{let}\,x=M\,\mathsf{in}\,N &~=~ \term{M;<x>.N}
\\     @_\val\,M\,N &~=~\term{N ; M ; <x>.x}
\end{aligned}
\] 
The embedding of $@_\val$ evaluates as illustrated below, where $N$ returns the value $V$ and $M$ returns $\lambda y.P$. First, $\term N$ and $\term M$ are evaluated, pushing $\term V$ and $\term{<y>.P}$ to the stack; then $\term{<y>.P}$ is popped, and executed with $\term V$ as its first argument on the stack, popping $\term V$ as $\term y$ and finally evaluating $\term{\{V/y\}P}$ to some value $\term W$. Note that this evaluates the argument before the function; the other way around is by the translation taking $@_\val\,M\,N$ to $\term{M ; N ; <y>.<x>.[y].x}$.
\[
	\term{N;M;<x>.x}:~~S \evs{N} S\,\term{V} \evs{M} S\,\term{V\,(<y>.P)} \ev{<x>} S\,\term{V} \evs{<y>} S \evs{\{V/y\}P} S\,\term{W}
\]

The sequential $\lambda$-calculus implicitly features the value/computation distinction of Levy's \cbpv~\cite{Levy-2003}, with the operand stack holding values. Making these explicit yields the variant of the calculus below, into which \cbpv\ embeds in the same way as the computational metalanguage. It distinguishes \emph{value} terms $\term V$ and \emph{computation} terms $\term M$, mediated by \emph{thunk} and \emph{force} constructs $\term{`!M}$ and $\term{`?V}$. The first-order calculus then arises by removing thunks $\term{`!M}$ as values, leaving only variables, and forced values $\term{`?V}$ as computations.
\[
	\term{V,W}~\coloneqq~\term{x}~\mid~\term{`!M} \qquad\qquad \term{M,N}~\coloneqq~\term{`?V}~\mid~\term{<x>.M}~\mid~\term{[V].M}~\mid~\term*~\mid~\term{M;N}
\]

In the monadic setting and \cbpv, sequencing is modelled by the \emph{return}- and \emph{let}-constructs, which pass a single return value. The translation interprets these as \emph{skip} and \emph{sequence}, with the return value passed on the stack. All three calculi thus express essentially the same notion of \emph{sequentiality}. This is made clearer still with the monadic \emph{bind} notation of Haskell, which is interpreted directly as $M\bind N=\term{M;N}$. The sequential $\lambda$-calculus may thus be viewed as a generalization of \cbpv\ where sequencing passes the entire argument stack instead of a single value. This interpretation, that $\term *$ passes all return values rather than no values, is expressed in the reduction semantics by the following rules.
\[
\begin{array}{c@{\qquad\qquad}c}
	\term{([M].N);P}\rw\term{[M].(N;P)} 
  & \term{(<x>.M);N}\rw\term{<x>.(M;N)} \quad (x\notin\fv(\term N))
\\[5pt]
	R           \ev{[M]} R\,\term M \evs{N}        S \evs{P} T
  & R\,\term{P} \ev{<x>} R          \evs{\{P/x\}M} S \evs{N} T
\end{array}
\]
These implement the idea that sequencing and prefixing in push- and pop-actions express the same notion of sequentiality, in accordance with the operational semantics. Together with standard $\beta$-reduction $\term{[M].<x>.N}\rw\term{\{M/x\}N}$ and the rule $\term{*;M}\rw\term M$, the interpretation of a \emph{let}-redex then reduces as follows. This demonstrates how values pushed before $\term *$ are passed on to the next computation.
\[
	\mathsf{let}\,x=\mathsf{return}\,M\,\mathsf{in}\,N \quad = \quad \term{([M].*) ; <x>.N}~\rw~\term{[M].(* ; <x>.N)}~\rw~\term{[M].<x>.N}~\rw~\term{\{M/x\}N}
\]

Like the computational metalanguage and \cbpv, the sequential $\lambda$-calculus may be may be extended with effect operators as well as constants, primitive operations, exceptions, etc.\ to give a model of higher-order computation with effects. The FMC however takes a different approach. Instead of introducing new primitives for effects, the machine and the existing constructs are generalised in subtle and minimal ways, so that effects and control flow may be modelled while preserving confluent reduction and typeability. The two generalisations, \emph{locations} and \emph{control}, are explored next.


\section{Locations}
\label{sec:locations}

To capture effects, the machine is generalised from one operand stack to multiple, named stacks, indexed in a global set of \emph{locations} $A=\{\lambda,a,b,c,\dots\}$, with $\lambda$ indicating the original or \emph{default} stack. Pop- and push-actions are parameterised in a location to operate on the designated stack, as $\term{a<x>.M}$ and $\term{[N]a.M}$. The default location $\lambda$ will be omitted from the notation to retain the previous $\term{<x>.M}$ and $\term{[N].M}$. This constitutes the FMC as previously published~\cite{Heijltjes-2022} (modulo the choice of primitives for sequencing).
\[
	\term{M,N} ~\coloneqq~ \term x ~\mid~ \term{a<x>.M} ~\mid~ \term{[N]a.M} ~\mid~ \term{*} ~\mid~ \term{M;N}
\]
Terms now operate on a family of stacks $S_A=\{S_a\mid a\in A\}$ called a \emph{memory}, where only a finite number of stacks are non-empty. Since a given term uses only a fixed, finite number of locations, in each particular case $A$ itself may be considered finite. For easier manipulation, a memory will be written as a sequence of terms indexed by locations, $a(\term{M})$, where terms on different locations may permute:
\[
	S_A ~\coloneqq~ a_1(\term{M_1}) \dots a_n(\term{M_n}) \qquad\text{where}~~a(\term M)\,b(\term N) = b(\term N)\,a(\term M)~~\text{if}~~a\neq b
\]
Pushing $\term M$ to the stack $S_a$ in the memory $S_A$ is then written $S_A\,a(\term M)$. The operational semantics of the indexed push- and pop-actions is as follows.
\[
	\term{[N]a.M}:~~S_A             \ev{[N]a} S_A\,a(\term N) \evs{M}        T_A \qquad\qquad
	\term{a<x>.M}:~~S_A\,a(\term N) \ev{a<x>} S_A             \evs{\{N/x\}M} T_A
\]
The reduction semantics expresses the independence of stacks on different locations by the following two rules. The first is the regular $\beta$-step from $\lambda$-calculus, when a push meets a pop on the same location. The second, \emph{passage} rule gives the case when the locations are distinct, resolved by permuting both operations past each other, in search of a counterpart that does match. (Note that the positioning of the location label indicates on which side the operation interacts.)
\[
	\term{[N]a.a<x>.M}~\rw~\term{\{N/x\}M} \qquad\qquad \term{[N]b.a<x>.M}~\rw~\term{a<x>.[N]b.M}\quad (a\neq b,\,x\notin\fv(\term N))
\]
Instead of the above two rules, the original presentation of the FMC~\cite{Heijltjes-2022} featured a single rule that allowed matching push- and pop-actions to interact \emph{at a distance}, ignoring other actions in between. The present formulation is now preferred: the rules are simpler, they give better normal forms where pop-actions precede push-actions, and they give a better equational theory, as follows. With $\eta$-equivalence $\term{M}\sim\term{a<x>.[x]a.M}$ where $x\notin\fv(\term M)$, reduction equivalence gives the following equations,  completing the intuitive idea that push- and pop-actions on distinct stacks are interchangeable.
\[
 	\term{a<x>.b<y>.M}\sim\term{b<y>.a<x>.M} \qquad\qquad \term{[P]a.[N]b.M}\sim\term{[N]b.[P]a.M}
\]
This equational theory then captures the algebraic theory for global store of Plotkin and Power~\cite{Plotkin-Power-2002}, via the encoding of store given below (see~\cite{Heijltjes-2022}).

Typing generalises along the same lines as memories, from one to many type vectors for input and output, indexed in the set of locations $A$. A \emph{memory type}, written $\type{!!t}$, is a family of type vectors in $A$, and like a memory, may be written as a sequence modulo permutation on distinct locations, as below. A type is then an implication between memory types, $\type{!!s=>!!t}$. The empty memory type is written $\type\e$, and as with terms, the main location $\lambda$ may be omitted, so that the notation is conservative over that for the sequential $\lambda$-calculus.
\[
	\type{s,t}~\coloneqq~\type{!!s=>!!t}
\qquad\quad
	\type{!!t}~\coloneqq~\type{a_1(t_1)..a_n(t_n)} 
\quad\text{where}~~
	\type{a(s)\,b(t)} = \type{b(t)\,a(s)}
~~\text{if}~~
	a\neq b
\]


\subsection{Modelling effects}

Stacks model higher-order store, input/output, and probablistic choice as follows.
\[
\begin{array}{@{}r@{}r@{}l@{\qquad}r@{}r@{}l@{}}
	c \smallbin{:=} M ~=~& \term{M ; <x>.c<\_>.[x]c.*} &\,:\, \type{c(t) => c(t)}
&   \mathsf{print}\,M ~=~& \term{M ; <x>.[x]\stdout.*} &\,:\, \type{\e => \stdout(t)} 
\\  !c                ~=~& \term{c<x>.[x]c.[x].*}      &\,:\, \type{c(t) => c(t)\,t} 
&       \mathsf{read} ~=~& \term{\stdin<x>.[x].*}      &\,:\, \type{\stdin(t) => t}  
\\&&& \mathsf{sample} ~=~& \term{\rnd<x>.[x].*}        &\,:\, \type{\rnd(t) => t}
\end{array}
\]
For higher-order store, a mutable variable $c$ is represented by a location, with \emph{update} $c \smallbin{:=} M$ and \emph{lookup} $!c$ as below left. Updating evaluates $\term M$, leaving the result $\term N$ on the main stack; this is popped as $\term x$; then $c$ is cleared by popping and discarding the value $\term P$, where the underscore $(\term{\_})$ represents a variable that may not occur; and finally $\term N$ (substituted for $\term x$) is pushed to $c$. 
\[
	S_A\,c(\term{P}) \evs{M} S_A\,c(\term P)\,\lambda(\term N) \ev{<x>} S_A\,c(\term P) \ev{c<\_>} S_A \ev{[N]c} S_A\,c(\term N)
\]
Lookup $!c$ pops the value from $c$ as $\term x$; restores it by pushing; and returns it on the main stack. Note that these encodings are familiar from Concurrent Haskell's MVars~\cite{PeytonJones-Gordon-Finne-1996}. The stack for $c$ is assumed to hold at most one term, but this need not be enforced explicitly since the encoding of the operations preserves it. 

Input/output is modelled by two  dedicated locations $\stdin$ and $\stdout$, with \emph{read} and \emph{print} operations as above. The $\stdin$ stack should allow only to \emph{pop} but not \emph{push}, and $\stdout$ only to push, which is again maintained by the encoding. Probabilistic choice is similar to input, with a location $\rnd$ representing a random number generator. The embeddings give a call--by--value semantics, returning a value to the stack in the case of lookup, read, and sample, and in the case of update and print first evaluating the argument $\term{M:\e=>t}$, which leaves a single value of type $\type t$ on the main stack. 

Untyped, the FMC provides an operational semantics to effect operators, and is able to express their call--by--name as well as their call--by--value interpretation. The simple types imposed by the translation are an innovation of the FMC, and agree with known \emph{intersection types} for store~\cite{Davies-Pfenning-2000,deLiguoro-Treglia-2021,Heijltjes-2025}. Typing for store is further similar to the state monad $T(\typ{t})=\typ{s}\smallbin\to(\typ{s}\smallbin\times\typ{t})$, but parameterized in a location $c$, with the type for lookup as $\type{c(s)=>c(s)\,t}$ and that for update as $\type{c(s)=>c(s)}$, corresponding to $T(1)=\typ{s}\smallbin\to(\typ{s}\smallbin\times1)$. Notable however is that extending a computation with a store $c$ in the FMC is not given by a transformation from $\type t$ to $\type{c(s)=>c(s)\,t}$, but from $\type{!!r=>!!t}$ to $\type{c(s)\,!!r=>!!t\,c(s)}$. It is thus not a state monad in the traditional sense, which is one reason why the FMC avoids the compositionality problem of the monadic approach. A final observation here is that for higher-order store, typing mutable variables themselves, as $c:\typ{s}$, does not guarantee termination, since it allows Landin's Knot~\cite{Landin-1964}, the fixed-point combinator $c\,\smallbin{:=} F(!c)\,;\,!c$, to be typed. The above type systems (intersection types, the state monad, and FMC typing) avoid this by typing both the input and output effect of store operations, which does enforce termination properties.

The locations of the FMC are akin to the \emph{channels} of process calculi~\cite{Milner-Parrow-Walker-1992}, as demonstrated by the encoding of input/output, and indeed channels are used in concurrency theory to encode store~\cite{Hirschkoff-Prebet-Sangiorgi-2020,Rocha-Caires-2021}. Future work will consider locations as channels for a message-passing concurrent FMC. Locations are also similar to the $\mu$-variables of $\lambda\mu$-calculus~\cite{Parigot-1992}, which likewise are names for stacks~\cite{Streicher-Reus-1998}. However, $\lambda\mu$-calculus abstracts over these, reifing stacks as values, a sufficiently different mechanism to the FMC that neither calculus readily encodes the other.


\section{Control}
\label{sec:control}

The contribution of the present paper is to extend the FMC to branching and looping computation, enabling the embedding of control flow operations including conditionals, exception handling, and loops. \emph{Skip} $(\term*)$, signifying successful termination, is generalized to a set of \emph{choices} $\{\trm*,i,j,k,\dots\}$ which name the potential branches of a computation, similar to \emph{exceptions} or \emph{exit codes}. Sequential composition $(\term{M;N})$ is generalized accordingly, to a \emph{case} $\term{M;i->N}$ which composes only on the chosen branch $i$: it first evaluates $\term M$, and if this exits with $i$ it continues with $\term N$, discarding it otherwise. Again for syntactic conservativity, sequencing is defined as $\term{M;N}=\term{M;*->N}$. Finally, a \emph{loop} construct $\term{M^i}$ is introduced, which repeats $\term M$ for as long as it exits with $i$, and terminates otherwise. The full calculus is then as follows.
\[
	\term{M,N}
~\coloneqq~ \term x
~\mid~		\term{[N]a.M}
~\mid~		\term{a<x>.M}
~\mid~		\term i
~\mid~		\term{M;i->N}
~\mid~		\term{M^i}
\]
In the operational semantics computations return a memory together with a choice. \emph{Choice}, \emph{case}, and \emph{loop} terms evaluate as follows, where $i\neq j$. 
\[
	\term i:~~S_A \evs[24]{i} S_A,\term i  
\qquad
	\term{M;i->N}:
\left\{\begin{array}{l}
	R_A \evs[24]{M} S_A,\term i \evs[24]{N} T_A,\term k
\\	R_A \evs[24]{M} S_A,\term j
\end{array}\right.
\qquad
	\term{M^i}:
\left\{\begin{array}{l}
	R_A \evs[24]{M} S_A,\term i \evs[24]{M^i} T_A,\term k
\\	R_A \evs[24]{M} S_A,\term j
\end{array}\right.
\]


\subsection{Expressing control flow}

Constructions for control flow are embedded as follows. Below left, the \emph{choice} and \emph{case} terms themselves are the standard \emph{throw} and \emph{try/catch} of exception handling, with an exception $e$ as a choice label. Below right, the boolean constants are encoded as choice terms pushed onto the stack. A conditional evaluates the condition, pops the (expected) boolean from the stack, and uses it as a choice to select the matching branch. Note that \emph{case} associates left, $\term{M;i->N;j->P}=\term{(M;i->N);j->P}$.
\[
\begin{aligned}
                \Throw\,e &~=~ \term{e}       &\qquad&& \top\,,\,\bot &~=~ \term{[\top].*}~,~\term{[\bot].*}
\\	\Try\,M\,\Catch\,e\,N &~=~ \term{M;e->N}  &\qquad&& \If\,B\,\Then\,M\,\Else\,N &~=~ \term{B\,;\,<x>.x\,;\,\top -> M\,;\,\bot -> N}
\end{aligned}
\]
The encoding of the booleans gives the general pattern for modelling constants and primitive operations. Primitive data such as bounded integers are modelled by a finite set of choices $\{c_1,\dots,c_n\}$, each representing a constant. A constant as an expression is encoded by pushing the choice to the stack, and functions on contants such as addition and multiplication are constructed from case switches, interpreted as below.
\[
	c_i ~=~ \term{[c_i].*} 
	\qquad 
	\Case\,M\,\Of\,\{ c_1\mapsto N_1,\dots,c_n\mapsto N_n\} ~=~ \term{M ; <x>.x ; c_1 -> N_1 ; \dots ; c_n -> N_n}
\]
These interpretations impose a call--by--value semantics on the encoded constructs, which follows the principle that \emph{exceptions} are choices as \emph{computations}, while \emph{constants} are choices as \emph{values}. The operational behaviour of the encoding is correct as long as constants and exceptions are modelled by distinct choice labels; then in the above interpretation of the case switch, an exeption in some $\term{N_k}$ cannot match a later constant $\term{c_m}$. 

The \emph{loop} construct $\term{M^i}$ encodes iteration, as expressed by its operational semantics and its rewrite rule, $\term{M^i}\rw\term{M;i->M^i}$. The use of the \emph{case} construct to model iteration means that if $\term{M}$ exits with any other choice $j\neq i$ this terminates the loop. Exceptions as choices in $\term M$ then behave as expected. Escape clauses such as \emph{break} and \emph{return} may likewise be modelled by choices, to be caught outside the loop. Where $\term{M^i}$ is naturally a \emph{``do~$\term M$~while~$\term i$''} loop, the escape mechanism allows easy \emph{while-do} loops as well. In the encodings below right, the \emph{do--while} loop repeats for the $\top$-choice of the boolean $\term B$, while the \emph{while--do} loop repeats if $\term M$ terminates correctly with $\term*$. Both catch the $\bot$-choice of $\term B$ and a $\Break$-choice as correct loop escapes (but omit to catch a \emph{return} clause).
\[
\begin{array}{@{}r@{}l@{\qquad}r@{}r@{}}
	\Break     &~=~ \term{\Break}      & \Do~M~\While~B ~=~& \term{(M\,;\,B\,;\,<x>.x)^\top\,;\,\bot->*\,;\,\Break->*}
\\	\Return~M  &~=~ \term{[M].\Return} & \While~B~\Do~M ~=~& \term{(B\,;\,<x>.x\,;\,\top -> M)^\star\,;\,\bot->*\,;\,\Break->*}
\end{array}
\]


\subsection{Control types}

The static semantics of the new constructs is that of a choice-indexed sum: computation may follow a finite number of possible branches, each labelled with a choice and returning a different memory. Typing is adjusted by changing output types to sums of memory types, as follows. The notation mirrors the syntax of terms ending in a choice, and sum types are considered moduly symmetry.
\[
	\type{s,t} ~\coloneqq~ \type{!!s => !!tI} \qquad\quad \type{!!tI} ~\coloneqq~ \type{!!t_{i_1}.i_1 + .. + !!t_{i_n}.i_n} \quad\text{where}\quad I=\{i_1,\dots,i_n\} 
\]
The meaning of a typing judgement $\term{M:!!s => !!tI}$ is that given an input memory $S_A:\type{!!s}$ evaluation returns a memory $T_A:\type{!!t_i}$ for some $i\in I$ --- \emph{if} it terminates, which is guaranteed in the absence of the loop construct, but not in its presence. Types thus guarantee \emph{progress} in the general case (Proposition~\ref{prop:progress}) and \emph{termination} in the loop-free case (Theorem~\ref{thm:termination}). The interpretation of types supports the following two properties. Below left is \emph{sum expansion}: a term with return type $\type{!!s_I}$ may also be typed with any larger return type $\type{!!sI+!!tJ}$, where the sum notation implicitly assumes $I\cap J=\varnothing$. Below right is \emph{stack expansion}: as before, a term that evaluates with a memory $\type{!!r}$ may also use a larger memory $\type{!!s\,!!r}$, leaving the additional memory $\type{!!s}$ unchanged in each branch of the computation. Implementing this, the notation $\type{!!s\,!!tI}$ prefixes the memory type $\type{!!s}$ to each $\type{!!t_i}$ in $\type{!!tI}$, taking the summand $\type{!!t_i.i}$ to $\type{!!s\,!!t_i.i}$. Since stacks are products this represents the familiar distribution law $A\times (B+C)=(A\times B)+(A\times C)$.
\[
\term{M:!!r=>!!sI} ~\Longrightarrow~ \term{M:!!r=>!!sI+!!tJ}
\qquad\qquad
\term{M:!!r=>!!tI} ~\Longrightarrow~ \term{M:!!r\,!!s => !!s\,!!tI}
\]
The constructors are typed as follows. Choice terms are injections, typed $\term{i: !!s => !!s.i + !!tJ}$, a type which may be obtained from $\term{i: \e=>\e.i}$ by stack and sum expansion. A case $\term{M;i->N}$ is typed as follows:
\[
 	\term{M: !!r => !!tJ + !!s.i} 
	\quad\text{and}\quad
	\left\{\begin{array}{ll}
	     \term{N: !!s => !!tJ}         & \text{or} 
	  \\ \term{N: !!s => !!tJ + !!u.i} 
	\end{array}\right\}
	\quad\Longrightarrow\quad
	\left\{\begin{array}{ll} 
	    \term{M;i->N : !!r => !!tJ}      & \text{or}
	 \\ \term{M;i->N : !!r => !!tJ + !!u.i}
	\end{array}\right.
\]
That is, $\term M$ must exit on $i$ with an appropriate memory $\type{!!s}$ to serve as input for $\term N$; for any other choice $j\in J$ the terms $\term M$ and $\term N$ must exit with the same type $\type{!!t_j}$, both matching a potential later case $\term{j->P}$; and $\term N$ may or may not exit with a choice $i$ and type $\type{!!u}$. The merger of both types $\type{!!tJ}$ represents the co-diagonal of the sum, $A+A\to A$. 

Composition is made flexible by stack and sum expansion, as follows. First, evaluating $\term M$ may provide fewer or more arguments on the stack than used by $\term N$, with the difference made up with stack expansion for $\term M$ or $\term N$ respectively. Second, the return choices for $\term M$ and $\term N$ need not coincide, or even overlap: as long as their types agree for the choices they have in common, they can be sum-expanded to match.

A loop $\term{M^i}$, which repeats on $i$, expects a type $\term{M:!!s=>!!s.i + !!tJ}$ where the input type coincides with the output type for $i$. The loop itself is then typed $\term{M^i:!!s => !!tJ}$, omitting the choice $i$, as it exits on the choices in $J$. This gives the expected typing pattern of iteration, which is to take $f:A\to A+B$ to $\mathsf{iter}\,f:A\to B$~\cite{Bloom-Esik-1993}. An interesting consequence is what happens with iteration on a term with only a single branch. This results in a sum type where $J$ is empty, the empty sum or \emph{void} type $\type{0}$. This is not ruled out in principle, and may happen for example with the loop $\term{i^i:\e => 0}$ for $\term{i:\e=>\e.i}$. Types $\type{!!s=>0}$ returning void are thus inhabited by terms looping on their only output branch, which are guaranteed non-terminating.

The type system imposes typing on the embedded constructions, including exceptions, constants, case switches, and while-loops. The type for booleans may be defined as $\type\B=\type{\e => \e.\bot + \e.\top}$, which correctly gives $\term{\bot:\B}$ and $\term{\top:\B}$. The type for the conditional is built up as follows, for a simple case where $\term B$ returns only a boolean and $\term M$ and $\term N$ return a single value of type $\type t$. The term $\term B$ returns a boolean on the main stack $\lambda$ for the default choice $\star$. This is picked up and executed by $\term{<x>.x}$, giving $\term{B;<x>.x}$ the type $\type\B$. Then the case $\term{\top -> M}$ composes on the $\top$-branch, returning $\type t$ on $\star$ instead, and $\term{\bot->N}$ composes on the $\bot$-branch, also returning $\type t$ for $\star$ and merging both branches.
\[
	\term{B : \e => \B.*} \qquad
	\term{<x>.x : \B => \e.\bot + \e.\top} \qquad
    \term{M : \e => t.*} \qquad
    \term{N : \e => t.*}
\]
\[
	\If\,B\,\Then\,M\,\Else\,N~=~\term{B ; <x>.x ; \top -> M ; \bot -> N : \e => t.*}
\]
Types for constants such as bounded integers use this general pattern, as $\type{\e => \e.i_1 + .. + \e.i_n}$ for constants $\{i_1,\dots,i_n\}$. Continuing the above example, if $B$ may throw an exception $e$, with the type $\type{\e => \B.* + \e.e}$, or one or both of $M$ or $N$ with type $\type{\e => t.* + \e.e}$, this is seamlessly integrated into the types to give the conditional the type $\type{\e => t.* + \e.e}$ as well. The following example will demonstrate how store integrates with these constructs, and how a \emph{while-do} loop is typed.

\begin{example}
Below is the interpretation and typing of the following imperative program.
\[
	\While~!a<5~\Do~a\,\smallbin{:=}\,!a+1
\]
It is broken up as $\While~B~\Do~N$ where $B={}!a<5$, $N=a\smallbin{:=}M$ and $M={}!a+1$. Let $\type\Z$ be a type for bounded integers and assume terms for addition $\term{`+: \Z\,\Z => \Z.*}$ and inequality $\term{`<: \Z\,\Z => \B.*}$ modelled by case switches. The interpretation of an expression $P+Q$ is its translation into reverse Polish notation, $\term{P;Q;`+}$, standard for a stack machine implementation. Then $M$ is interpreted as below, where $\term{[1].*:\e=>\Z.*}$ is stack-expanded with the memory type $\type{a(\Z)\,\Z}$ to allow the composition, and $\term{`+: \Z\,\Z => \Z.*}$ with $\type{a(\Z)}$. Note that this implicit expansion is the equivalent of lifting into the state monad in the monadic approach.
\[
\term{a<x>.[x]a.[x].* : a(\Z)=>a(\Z)\,\Z.*} \qquad
          \term{[1].* : \Z\,a(\Z)=>a(\Z)\,\Z\,\Z.*} \qquad
             \term{`+ : \Z\,\Z\,a(\Z)=>a(\Z)\,\Z.*}
\]
\[
 	M~=~!a + 1~=~\term{a<x>.[x]a.[x].* ; [1].* ; `+ : a(\Z)=>a(\Z)\,\Z.*}
\]
Next, the term $\term N$ composes this with inputs on $a$ and $\lambda$, returning on $a$, and $B$ is similar to $M$.
\[
\begin{array}{r@{}r@{}r@{}l}
   N ~=~ & a\,\smallbin{:=}\,!a + 1 ~=~ & \term{M ; <y>.a<\_>.[y]a.*}                      &~:~ \type{a(\Z)=>a(\Z).*}
\\ B ~=~ &	                 !a < 5 ~=~ & \term{a<x>.[x]a.[x].* ; [5].* ; `<}            &~:~ \type{a(\Z)=>a(\Z)\,\B.*}
\end{array}
\]
The loop $\While~B~\Do~N$ translates to $\term{(B;<x>.x;\top->N)^\star;\bot->*}$, which is built up below. First, to compose with $\term B$, the term $\term{<x>.x:\B => \e.\bot+\e.\top}$ is lifted with $\type{a(\Z)}$. The second line composes $\term{B}$ with $\term{\top -> N}$, and the third line gives the loop on the default choice $\star$, which is well typed since the output type $\type{a(\Z)}$ matches the input. The final step adds the catching of $\bot$, giving the final type of the program as an update on $a$.
\[
\begin{array}{r@{}l@{}l}
                                    & \phantom{\term{(B;}}\term{<x>.x}                   &~:~ \type{\B\,a(\Z)=>a(\Z).\bot + a(\Z).\top}
\\                                  & \phantom( \term{B; <x>.x ; \top->N}                &~:~ \type{a(\Z)=>a(\Z).\bot + a(\Z).*}
\\                                           & \term{(B; <x>.x ; \top->N)^\star}         &~:~ \type{a(\Z)=>a(\Z).\bot}
\\ \While~!a<5~\Do~a\,\smallbin{:=}\,!a+1 ~=~& \term{(B; <x>.x ; \top->N)^\star;\bot->*} &~:~ \type{a(\Z)=>a(\Z).*}
\end{array}
\]
\end{example}

The above demonstrates how the type system of the FMC carries over to the embedded imperative programming language, complete with higher-order store, input/output, probabilistic choice, constants, conditionals, exceptions, and iteration with escape. The properties of the type system, \emph{progress} (Proposition~\ref{prop:progress}) and \emph{termination} in the absence of loops (Theorem~\ref{thm:termination}), apply. Concretely, since types make all potential branches of a computation explicit, \emph{progress} guarantees that exceptions are caught and loop escapes handled correctly. The embedding is naturally compatible with that of the call--by--value $\lambda$-calculus, which likewise returns values to the main stack, giving the embedding of an ML-like language.

As a final consideration it will be shown how algebraic (non-inductive) datatypes with a call--by--name semantics embed. With the \cbn~$\lambda$-calculus and recursive definitions, the latter omitted here, this forms the core of the Haskell programming language. Data constructors embed as choice terms, but not returned on the stack as in the call--by--value interpretation. Following the standard \cbn\ interpretation, a fully applied constructor $i$ thus leaves its data on the stack and exits with choice $i$, as below. A case switch on a datatype consists of a series of \emph{pattern matching} cases $i\,x_1\,\dots\,x_m\to N$ which bind the arguments $M_k$ of the datatype to the variables $x_k$ for use in $N$; the notation below uses vectors of variables $\trm{!x}$ for the parameters of each case. In the FMC interpretation the arguments are passed on the stack, so that pattern-matching may be given by abstractions, where the notation uses $\term{<!x>}$ for $\term{<x_1>..<x_n>}$. The computation $\term{N_i}$ for the matching case is then pushed to the stack, since evaluating it might incorrectly trigger a later case, and is popped and evaluated by $\term{<x>.x}$ after the case switch is completed.
\[
\begin{aligned}
                                                       i\,M_1\,\dots\,M_m  &~=~  \term{[M_m]\dots[M_1].i}
\\ \Case\,M\,\Of\,\{ i_1\,\rvec x_1 \to N_1,\dots,i_n\,\rvec x_n\to N_n\}  &~=~  \term{M\,;\,i_1 -> <!x_1>.[N_1].*\,;\,\dots\,;\,i_n -> <!x_n>.[N_n].*\,;\,<x>.x}
\end{aligned}
\]
Typing is as follows. A datatype definition as below left, using Haskell's \texttt{data} keyword, creates a new type $\sigma$ as a sum--of--products indexed by constructors. The FMC interprets the type $\sigma$ directly as a sum type.
\[
	\texttt{data}~\typ s~=~i_1\,\typ{!t_1}\mid\dots\mid i_n\,\typ{!t_n} \qquad\Longrightarrow\qquad \type{s}~=~\type{\e => !t_1.i_1 + .. + !t_n.i_n}
\]
The embedding of datatypes is then a typed embedding, extending that of the call--by--name $\lambda$-calculus. A typed constructor is interpreted as below, where $i$ is some $i_k$ of the above datatype with $\typ{!t_k}=\typ{t_1..t_m}$, so that it has type $\type s$ when fully applied. 
\[
	i: \typ{t_1 \to .. \to t_m \to s} ~=~  \term{i: t_1 .. t_m => !t_1.i_1 + .. + !t_n.i_n}
\qquad
	\term{[M_m]\dots[M_1].i : s}
\]
The case switch on $M:\typ s$ requires all terms $N_i$ to share the same type $\typ r$, and will then return $\typ r$ itself. The overall type $\type r$ for the case switch is composed in the following way.
\[
\term{M:s} \qquad \term{!x_i:!t_i} \qquad \term{N_i:r} \qquad \term{<!x_i>.[N_i].* : !t_i => r.*}
\]
\[
\term{M\,;\,i_1 -> <!x_1>.[N_1].*\,;\,\dots\,;\,i_n -> <!x_n>.[N_n].*\,;\,<x>.x : r}
\]


\subsection{Discussion and related work}

The present discussion will compare with related type systems and semantics for branching computation, in particular with typed exceptions. These have been studied in various forms, with subtle syntactic and semantic differences that the discussion will attempt to clarify.

The FMC takes the \emph{exceptions--as--values} approach~\cite{Appel-MacQueen-Miler-Tofte-1988,PeytonJones-Reid-Henderson-Hoare-Marlow-1999,Spivey-1990,Wadler-1985}, which follows the idea that computations return either a value or an exception. This casts exceptions as \emph{coproducts}, characterized by the exception monad $TX=E+X$ for a set of exceptions $E$. Going back to Gentzen's natural deduction~\cite{Gentzen-1935}, coproducts in $\lambda$-calculus have been handled with a case distinction on each summand, as below. 
\[
	\infer{\Case\,M\,\Of\,\{\mathsf{inl}\,x\to N,\mathsf{inr}\,y\to P\}: C}{
	    M: A+B 
	 && \begin{array}[b]{@{}c@{}} [x:A] \\[-3pt] \vdots \\[-3pt] N:C \end{array}
	 && \begin{array}[b]{@{}c@{}} [y:B] \\[-3pt] \vdots \\[-3pt] P:C \end{array}
	}
\]
Benton and Kennedy's \emph{exception handlers}~\cite{Benton-Kennedy-2001} adapt case switches to be \emph{fall-through}, capturing the return value and some (but not necessarily all) exceptions, with uncaught exceptions being passed through. This appears crucial to a compositional treatment of exceptions, by macro expansion. Otherwise, the encoding of a catch-construct $\Try\,M\,\Catch\,e\,N$ depends on the potential exceptions $e_1\dots e_n$ that $M$ might throw:
\[
	\Case\,M\,\Of\,\{\Return~x \to \Return~x,\,e_1\to e_1,\,\dots,\,e_n\to e_n,\,e\to N\}
\] 

Exception handlers feature a clear notion of \emph{sequentiality}, and are implemented through a \emph{continuation stack} similar to the one that will be used for the FMC in Definition~\ref{def:machine}. They are introduced into \cbpv\ in~\cite{Levy-2006-Exceptions}, and later generalized by Plotkin and Pretnar to \emph{effect handlers}~\cite{Plotkin-Pretnar-2013}, a general model of branching computation that will be discussed below.

The FMC features two key adaptations to the above. First, pattern-matching is avoided by passing values on the stack. Second, by unifying \emph{sequencing} and \emph{exceptions}, case switches may become not only \emph{fall-through} but also \emph{single-case}, since multiple cases can now be handled sequentially. Both changes are made possible, at least conceptually, by the model of the FMC as a sequential stack-based calculus.

A different approach to exceptions uses classical \emph{continuations}~\cite{Bierman-1998,Crolard-1999,DeGroote-1995,Kameyama-Sato-2002,Lebresne-2009,Ong-Stewart-1997,VanBakel-2019}, which give a typed $\lambda$-calculus with continuation operators such as call/cc~\cite{Griffin-1990} via the correspondence with classical logic. Continuation operators have a natural interpretation also on the Krivine machine~\cite{Carraro-Ehrhard-Salibra-2012,Streicher-Reus-1998}, where they \emph{reify} the stack, allowing to pass an entire stack as an argument. This is a different mechanism to how the FMC extends the KAM, not easily simulated in either direction. From an operational perspective, exceptions and continuations are incomparable constructs: neither can macro-express the other, not even in the presence of types or state~\cite{Laird-2002,Riecke-Thielecke-1999}. 

A separate but related notion are \emph{jumps}, which like the present \emph{control} extension, model branching and looping computation. Low-level languages, where jumps are a basic construct, may be typed~\cite{Saabas-Uustalu-2007}. In a more general setting, Fiore and Staton model jumps through \emph{explicit substitutions}~\cite{Fiore-Staton-2014}, and Maurer, Downen, Ariola, and Peyton Jones introduce jumps as explicit substitutions to optimize case switches in the Haskell compiler~\cite{Maurer-Downen-Ariola-PeytonJones-2017}. Syntactically, an explicit substitution $M[N/x]$ is very similar to a \emph{case} $\term{M;i->N}$, with the only difference that the former uses a \emph{variable} and the latter a \emph{choice} (a constant). Maurer et al.\ demonstrate their appeal as a \emph{single-case} construct, and further allow loops. This was a major inspiration for the present approach, and earlier drafts of this work adopted their terminology of \emph{jump} and \emph{join}~\cite{Heijltjes-2024}, though this was changed to \emph{choice} and \emph{case} for the following considerations.

There are key distinctions between explicit substitutions, jumps, and exceptions. First, the continuation $N$ in $M[N/x]$ may be duplicated, where both forward program jumps and exception handling are \emph{affine}: the continuation is used or discarded, but not duplicated. Fiore and Staton model jumps by making explicit substitutions affine through syntactic restrictions, as do Maurer et al., who restrict the target variables of an explicit substitution to \emph{head} position. Second, explicit substitutions and jumps are \emph{static}: the connection between the jump (the variable) and the continuation is determined by the term structure, before evaluation. Exceptions however are \emph{dynamic}: the choice whether a continuation is used or discarded is made only during evaluation. This is the key distinction between variables and constants, and is reflected also in the types for both constructs, where jumps and explicit substitutions are generally typed with type arrows as intuitionistic continuations, but exceptions as coproducts.

Plotkin and Pretnar's \emph{effect handlers}~\cite{Kammar-Lindley-Oury-2013,Marsik-Amblard-deGroote-2021,Plotkin-Pretnar-2013,Wu-Schrijvers-2015} are a general model of branching computation. Built over two layers, \cbpv\ as a model of sequential computation is extended with \emph{effect operator symbols}, that create symbolic branching points in the computation tree. A layer of \emph{effect handlers} then interprets these branching points operationally to model effects. This treatment of branching is sufficiently different from the present \emph{control} extension that it is not directly obvious which concept in one model maps onto which concept in the other. 

Following the origin of effect handlers in exception handlers, one view would be of operators as a generalisation of \emph{raising} an exception, analogous to an FMC \emph{choice}, and handlers as \emph{catching} it, analogous to \emph{case}. Handlers and cases have in common that they are matched \emph{dynamically}, but there are three key distinctions. First, handlers are generally defined globally, separate from computation terms, where the FMC has cases within the term language. Second, the standard \emph{deep} handlers apply to multiple operators in sequence, and thus are not \emph{affine}, though \emph{shallow} handlers are~\cite{Hillerstrom-Lindley-2018}. Third, like exception handlers, effect handlers require a return case as well as catching effects, and are thus \emph{fall-through} but not \emph{single-case}.

Taking a different perspective, like \cbpv\ the sequential $\lambda$-calculus is a model of higher-order sequential computation into which effect operators and handlers may be introduced. This is a present line of investigation, but preliminary observations suggest that this forms a natural combination: the presence of sequential composition as a primitive, enabled by passing arguments on the stack, means the use of continuations may be avoided. This not only simplifies the model but moreover allows the characteristic algebraic laws to be enforced by the syntax, rather than imposed externally. A further line of enquiry is how effect handlers relate to \emph{locations}, noting that both model specific effects, in particular \emph{store}. These observations serve to underline that both models, the FMC and handlers, complement each other in interesting ways. 

To summarise, the overall technical contributions are as follows. First, the new \emph{choice} and \emph{case} constructs of the FMC, in the termilogy used here, combine the \emph{single-case} nature of explicit substitutions with the \emph{affine}, \emph{dynamic} coproduct semantics of the exceptions--as--values paradigm. Second, inherited from the first iteration of the FMC, is the passing of values through (multiple, independent) stacks, which reduces pattern-matching to case-matching. Third, at a higher level, these two contributions allow the FMC to \emph{unify} (rather than \emph{combine}) the concepts of sequential composition, exception handling, and case switches, and seamlessly integrate them with the effects encoded through \emph{locations}. Fourth, again derived from the previous, is the type system, which gives a natural coproduct semantics to these constructions (to be confirmed formally in future work), where approaches through explicit substitution inherit continuation-style typing. Finally, the treatment of \emph{loops} with escape, enabled by the design of the choice and case constructs, appears novel.

The remaining sections will present the formal development of the calculus: its operational semantics, reduction relation, confluence, types, and machine termination and normalization of the typed calculus. Its denotational semantics will be the subject of future work.

The FMC as presented here covers a complete \emph{theoretical} functional--imperative language. In practice, further features are expected. \emph{Inductive data types} and \emph{recursion} require (standard) syntactic extensions but otherwise seem unproblematic. As the encoding of non-inductive data types above suggests, to allow recursive type aliasing is sufficient to make these inductive. An open question here is whether such a treatment may solve the problem of \emph{extensible data types}~\cite{Swierstra-2008}. \emph{Local store} would require a construct for introducing new locations, which appears similarly free of complications. \emph{First-class} locations, storing and passing on locations as arguments, would cover things like \emph{pointers} and \emph{file handles}, and appears more challenging to administer in the type system. Even more so \emph{arrays} and \emph{pointer arithmetic}, which call for dependent types (note that adding these untyped is trivial; the challenge again is preserving types and confluence).


\section{The Functional Machine Calculus}
\label{sec:FMC}

This section will present the new Functional Machine Calculus and its operational aspects: the abstract machine, a big-step evaluation relation, and the reduction relation. It will demonstrate that these agree: big-step evaluation defines complete runs on the machine (Proposition~\ref{prop:operational}), weak head reduction simulates the machine (Proposition~\ref{prop:weak-head}), and evaluation commutes with reduction (Proposition~\ref{prop:evaluation-reduction}). Confluence is proved in the next section.

Let $x,y,z$ range over \emph{variables}; $a,b,c$ over a global set of \emph{locations} $A$ with distinguished element $\lambda$; and $i,j,k$ over \emph{choices} with distinguished choice $\trm*$, with $I,J$ denoting finite sets of choices.

\begin{definition}
\emph{Terms} are given by the following grammar.
\[
	\term{M,N}
~\coloneqq~ \term x
~\mid~		\term{[N]a.M}
~\mid~		\term{a<x>.M}
~\mid~		\term i
~\mid~		\term{N;i->M}
~\mid~		\term{M^i}
\]
The constructs are: a \emph{variable} $\term x$, an \emph{application} or \emph{push} on location $a$, an \emph{abstraction} or \emph{pop} on $a$ that binds $\term x$ in $\term M$, a \emph{choice} $\term i$, a \emph{case} $\term{N;i->M}$, and a \emph{loop} $\term{M^i}$. The notions of \emph{variable binding}, \emph{free variables} $\fv(-)$, and \emph{capture-avoiding substitution} $\term{\{N/x\}M}$ of $\term N$ for $\term x$ in $\term M$ are standard.
\end{definition}

Define \emph{argument stacks} $S,T$ as stacks of terms, \emph{memories} $S_A$ as families of stacks in the set of locations $A$, and \emph{continuation stacks} $K,L$ as stacks of \emph{conditional continuations} $\term{i->M}$, as follows.
\[
				S,T	 ~\coloneqq~ \e ~\mid~ S\,\term{M}
\qquad\qquad	S_A  ~\coloneqq~ \{S_a\mid a\in A\}
\qquad\qquad	K,L  ~\coloneqq~ (\term{i->M})\,K~\mid~\e
\]
Stacks are composed by juxtaposition, $ST$, lifted to memories pointwise: $S_AT_A=\{S_aT_a\mid a\in A\}$. Write $a(\term{M})$ for the singleton memory with $\term M$ on the stack $a$ and the empty stack on other locations. Memories may be assumed to have \emph{finite support} (only finitely many stacks are non-empty), since a term uses only a fixed, finite set of locations. \emph{Streams} may be used informally instead of stacks to model certain effects.

\begin{definition}
\label{def:machine}
The \emph{abstract machine} is as follows. \emph{States} are triples $(S_A,\term M,K)$ of a memory $S_A$, term $\term M$, and continuation stack $K$. \emph{Transitions} are given by the top-to-bottom rules below, where $i\neq j$.
\[
\begin{array}{@{(\,}l@{\,,\,}r@{\,,\,}r@{\,)}}
         {S_A}                 & \term  {[N]a.M} & {K} 
\\\hline {S_A\msep a(\term N)} & \term       {M} & {K}
\\[10pt] {S_A\msep a(\term N)} & \term  {a<x>.M} & {K}
\\\hline {S_A}                 & \term{\{N/x\}M} & {K}
\end{array}
\qquad\qquad
\begin{array}{@{(\,}l@{\,,\,}r@{\,,\,}r@{\,)}}
         {S_A} & \term{i} & {(\term{i->M})\,K}
\\\hline {S_A} & \term{M} &                {K}
\\[10pt] {S_A} & \term{i} & {(\term{j->M})\,K}
\\\hline {S_A} & \term{i} &                {K}
\end{array}
\qquad\qquad
\begin{array}{@{(\,}l@{\,,\,}r@{\,,\,}r@{\,)}}
         {S_A} & \term{N;i->M} &                  {K}
\\\hline {S_A} & \term{N}      &   {(\term{i->M})\,K}
\\[10pt] {S_A} & \term{M^i}    &                  {K}
\\\hline {S_A} & \term{M}      & {(\term{i->M^i})\,K}
\end{array}
\]
A \emph{run} of the machine is a sequence of steps written with a double line as below. A \emph{final} state is of the form $(S_A,\term i,\e)$ and a \emph{failure} state of the form $(S_A,\term x,K)$, or $(S_A,\term{a<x>.M},K)$ where $S_a=\e$. A run is \emph{successful} if it terminates in a final state.
\[
	\steps {S_A}MK {T_A}NL
\]
\end{definition}

Observe that every state is either final, a failure state, or has a transition. The machine gives the small-step operational semantics of the FMC. The following big-step evaluation relation $(\evalarrow)$ describes the overall behaviour of successful runs of the machine, and formalizes the intuitive version of the introduction: 
\[
	S_A \evs{M} T_A,\term i \quad\iff\quad \evl{S_A}M{T_A}i
\]
Like the machine, it is deterministic, i.e.\ it is a partial function. 

\begin{definition}
The \emph{evaluation} relation $\evl{S_A}M{T_A}i$ is defined inductively as follows, where $i\neq j$.
\[
\bigstep{}{\eval {S_A} i {S_A} i}
\quad
\begin{array}{@{}r@{\quad~~}r@{\quad~~}r@{}}
	\bigstep  
	  { \eval {S_A\msep a(\term N)}     M  {T_A} i }
	  { \eval {S_A}               {[N]a.M} {T_A} i }
&
	\bigstep
	  { \eval {R_A}  M       {S_A} i \quad
	    \eval {S_A}  N       {T_A} j }
	  { \eval {R_A} {M;i->N} {T_A} j }
&
	\bigstep
	  { \eval {R_A}  M       {T_A} i }
	  { \eval {R_A} {M;j->N} {T_A} i }
\\ \\[-10pt]
	\bigstep
	  { \eval {S_A}               {\{N/x\}M}  {T_A} i }
	  { \eval {S_A\msep a(\term N)} {a<x>.M}  {T_A} i }
&
	\bigstep
	  { \eval {R_A}  M    {S_A} i \quad
	    \eval {S_A} {M^i} {T_A} j }
	  { \eval {R_A} {M^i} {T_A} j }
&
	\bigstep
	  { \eval {R_A}  M    {T_A} i }
	  { \eval {R_A} {M^j} {T_A} i }
\end{array}
\]
\end{definition}

\begin{proposition}
\label{prop:operational}
Small-step and big-step semantics agree:
\[
	\steps {S_A}M\e {T_A}i\e \quad\iff\quad \evl{S_A}M{T_A}i
\]
\end{proposition}

\begin{proof}
$\Longrightarrow$ By induction on the run of the machine. $\Longleftarrow$ By induction on $\evalarrow$.
\end{proof}

\begin{definition}
The \emph{reduction} relation $\rw$ is given by closing the following rules under any context, where $a\neq b$, $i\neq j$, in the \emph{passage} rule $x\notin\textsf{fv}(\term N)$, and in the \emph{prefix (pop)} rule $x\notin\textsf{fv}(\term M)$.
\[		
\begin{array}{@{}l@{\quad~~}r@{}l@{\qquad~~}r@{}l@{\quad~~}r@{}}
		\textnormal{beta}    & \term{[N]a.a<x>.M} &~\rw~\term{\{N/x\}M}    & \term{(a<x>.N);i->M} &~\rw~\term{a<x>.(N;i->M)} & \textnormal{prefix (pop)}
	\\	\textnormal{passage} & \term{[N]b.a<x>.M} &~\rw~\term{a<x>.[N]b.M} & \term{([P]a.N);i->M} &~\rw~\term{[P]a.(N;i->M)} & \textnormal{prefix (push)}
	\\	\textnormal{select}  & \term{i;i->M}      &~\rw~\term{M}           & \term{P;i->N;i->M}   &~\rw~\term{P;i->(N;i->M)} & \textnormal{associate}
	\\	\textnormal{reject}  & \term{i;j->M}      &~\rw~\term{i}		   & \term{M^i}			  &~\rw~\term{M;i-> M^i}     & \textnormal{unroll}    
\end{array}		
\]
\emph{Weak head reduction} $\rw_\wh$ is given by closing the reduction rules under application contexts only: if $\term M\rw_\wh\term N$ then $\term{[P]a.M}\rw_\wh\term{[P]a.N}$. The reflexive--transitive closure of a reduction relation $\rw$ is written $\rws$, and reduction to normal form $\rwn$.
\end{definition}

Weak head reduction operates under a sequence of applications corresponding to a memory on the machine. To relate the machine and reduction, define the \emph{readback} relation $(\mapsto)$ from states to terms by the exhaustive application of the following steps, reversing the \emph{push} and \emph{sequence} rules of the machine:
\[
			     (\e,\term M,\e) \mapsto \term M
\qquad (S_A\msep a(\term N),\term M,K) \mapsto (S_A,\term{[N]a.M},K)
\qquad (S_A,\term M,(\term{i->N})\,K) \mapsto (S_A,\term{M;i->N},K)
\]
The following then shows that weak head reduction simulates the machine.

\begin{proposition}
\label{prop:weak-head}
If a state $(S_A,\term M,K)$ reads back to $\term{M'}$ and evaluates to a state $(T_A,\term N,L)$, then the latter state reads back to a term $\term{N'}$ such that $\term{M'}\rws_\wh\term{N'}$.
\[
\begin{tikzpicture}[x=56pt,y=-30pt]
	\node (S) at (0,0) {$(S_A,\term M,K)$};
	\node (T) at (0,1) {$(T_A,\term N,L)$};
	\node (M) at (1,0) {$\term{M'}$};
	\node (N) at (1,1) {$\term{N'}$};
	\draw[|->]        (S)--(M);
	\draw[|->,dash pattern = on 2pt off 2pt] (T)--(N);
	\draw[->,rws,dash pattern = on 2pt off 2pt]     (M)--node[right]{$\scriptstyle\wh$} (N);
	\draw[double distance=1.5pt](-25pt,0.5)--(25pt,0.5);
\end{tikzpicture}
\]
\end{proposition}

\begin{proof}
By induction on the machine run.
\end{proof}

In the other direction, weak head reduction preserves and reflects the evaluation behaviour of terms.

\begin{proposition}
If $\term{M}\rw_\wh\term{N}$ then $\evl{S_A}M{T_A}i$ if and only if $\evl{S_A}M{T_A}i$.
\[
\begin{tikzpicture}[x = 30pt , y = -30pt]
	\node (M) at (0,0) {$S_A,\term M$};
	\node (N) at (2,0) {$S_A,\term N$};
	\node (T) at (1,1) {$T_A,\term i$};
	\node at (1,0)     {$\vphantom M\rw_\wh$};
	\node at (0.3,0.5) {\rotatebox[origin=c]{45}{$\evalarrow$}};
	\node at (1.7,0.5) {\rotatebox[origin=c]{-45}{$\evalarrow$}};
\end{tikzpicture}
\]
\end{proposition}

\begin{proof}
By induction on $\rw_\wh$.
\end{proof}

Combining both directions, evaluation with an empty memory $\evl\e M{S_A}i$ coincides with weak head reduction to a term of the form $\term{[N_1]a_1\dots[N_n]a_n.i}$ that represents the memory $S_A$ with the choice $i$.

The following establishes that reduction in any context commutes with evaluation, demonstrating, in essence, that the evaluation semantics and reduction semantics of the FMC are compatible. Reduction may then be viewed as compile-time optimization, preserving the behaviour of evaluation. To state this formally reduction $\rw$ is extended to memories: if $\term M\rw\term N$ then $S_A\,a(\term M)\rw S_A\,a(\term N)$, and if $S_A\rw T_A$ then $S_A\,a(\term M)\rw T_A\,a(\term M)$. Note that this does not extend the machine itself, but enables to compare evaluation of terms before and after reduction, modelling the optimization of stored functions.

\begin{proposition}
\label{prop:evaluation-reduction}
If $R_A\rws S_A$, $\term M\rws\term N$, and $\evl{R_A}M{T_A}i$ then there is a memory $U_A$ such that $T_A\rws U_A$ and $\evl{S_A}N{U_A}i$.
\[
\begin{tikzpicture}[x = 50pt , y = -30pt]
	\node (R) at (0,0) {$R_A,\term M$};
	\node (S) at (1,0) {$S_A,\term N$};
	\node (T) at (0,1) {$T_A,\term i$};
	\node (U) at (1,1) {$U_A,\term i$};
	\draw[->,rws]           (R)--(S);
	\draw[->,rws,dashed]    (T)--(U);
	\node at (0,.5) {$\evalarrow$};
	\node at (1,.5) {$\evalarrow$};
\end{tikzpicture}
\]
\end{proposition}

\begin{proof}
By induction on the measure $(m,n)$ where $m$ is the size of the derivation for $\evl{R_A}M{T_A}i$ and $n$ is the number of reduction steps in $\term M\rws\term N$, strengthening the statement with the assertion that the size of the derivation for $\evl{S_A}N{U_A}i$ is at most $m$. The proof is similar to that of Lemma~\ref{lem:eval-reduce}.
\end{proof}

The above propositions serve to demonstrate that the reduction relation is the correct one for the calculus, given its operational semantics. The reduction rules are \emph{sound} for evaluation, in the sense of Proposition~\ref{prop:evaluation-reduction}, and \emph{complete} in the sense that weak head reduction implements the machine, Proposition~\ref{prop:weak-head}. The next section will show that the reduction semantics is \emph{consistent} by demonstrating confluence.


\section{Confluence}
\label{sec:confluence}

The confluence proof follows the standard \emph{parallel reduction} technique~\cite{Takahashi-1995}. Reduction is split into \emph{duplicating} reduction $\rw_\dup$, comprising \emph{beta} and \emph{unroll}, and \emph{affine} reduction $\rw_\aff$, consisting of the remaining rules. The former is shown to be confluent by parallel reduction, while the latter is shown confluent and terminating by Newman's Lemma. The two relations are then shown to commute.

\begin{lemma}
\label{lem:aff}
Affine reduction $\rw_\aff$ is terminating and confluent.
\end{lemma}

\begin{proof}
For termination, define a measure on terms as the pair $(n,m)$ where:
\begin{itemize}
	\item $n$ is the sum over the size of $\term M$ for every subterm $\term{M;i->N}$, and
	\item $m$ is the sum over the size of $\term M$ for every subterm $\term{[N]a.M}$. 
\end{itemize}
The first component is invariant under the \emph{passage} rewrite rule, and strictly reduces for the other rules (\emph{select}, \emph{reject}, both \emph{prefix} rules, and \emph{associate}), while the second component strictly reduces for \emph{passage}, proving termination.

Confluence follows by Newman's Lemma from local confluence. There are the following critical pairs.

\begin{itemize}
	\item\emph{Passage--prefix}:
\[
	\term{[N]b.a<x>.M;i->P} \rw_\aff
	\left\{\begin{aligned}
		\term{Q}  &= \term{a<x>.[N]b.M;i->P}
	\\  \term{Q'} &= \term{[N]b.(a<x>.M;i->P)}
	\end{aligned}\right.
\]
This is closed as follows.
\[
	\left.\begin{aligned}
		\term{Q}  &\rw_\aff \term{a<x>.([N]b.M;i->P)} 
	\\  \term{Q'} &\rw_\aff \term{[N]b.a<x>.(M;i->P)}
	\end{aligned}\right\}
	\rw_\aff \term{a<x>.[N]b.(M;i->P)}
\]

	\item\emph{Select--associate}:
\[
	\term{i;i->M;i->N} \rw_\aff
	\left\{\begin{aligned}
		\term{P} &= \term{M;i->N}
	\\  \term{Q} &= \term{i;i->(M;i->N)}
	\end{aligned}\right.
\]
This is closed by $\term Q\rw_\aff\term P$. 

	\item\emph{Reject--associate}:
\[
	\term{i;j->M;j->N} \rw_\aff
	\left\{\begin{aligned}
		\term{P} &= \term{i;j->N}
	\\  \term{Q} &= \term{i;j->(M;j->N)}
	\end{aligned}\right.
\]
This is closed by $\term{P}\rw_\aff\term i$ and $\term{Q}\rw_\aff\term i$. 

	\item\emph{Prefix--associate}: The case for \emph{prefix (pop)} is as follows.
\[
	\term{a<x>.M;i->N;i->M}
	\left\{\begin{aligned}
		\term{Q}  &= \term{a<x>.(M ; i->N) ; i->P}
	\\  \term{Q'} &= \term{a<x>.M ; i->(N ; i->P)}
	\end{aligned}\right.
\]
This is closed as follows.
\[
	\left.\begin{aligned}
		\term{Q}  \rw_\aff \term{a<x>.(M ; i->N ; i->P)} 
	\\  \term{Q'} 
	\end{aligned}\right\}
	\rw_\aff \term{a<x>.(M ; i->(N ; i->P))}
\]
The case for \emph{prefix (push)} is similar.
\end{itemize}
\end{proof}

For duplicating reduction $\rw_\dup$, parallel reduction is defined by marking selected redexes and reducing these simultaneously by induction on the term.

\begin{definition}
A \emph{marked} term is one equipped with a marking on a selection of \emph{beta-} and \emph{unroll-}redexes, indicated by underlining. The \emph{marked reduct} $\term{|_M_|}$ of a marked term $\term M$ is defined inductively as follows.
\[
\begin{array}{@{}r@{{}={}}l@{\quad~~}r@{{}={}}l@{\quad~~}r@{{}={}}l@{\quad~~}r@{{}={}}l@{}}
	\term{|_ __{[N]a.a<x>}.M _|} & \term{\{|_N_|/x\}|_M_|}     
  & \term{|_x_|}                 & \term x 
  & \term{|_ [N]a.M _|}          & \term{[|_N_|]a.|_M_|} 
  & \term{|_ M;i->N _|}          & \term{|_M_|;i->|_N_|}
\\  \term{|_ __{M^i}_|}          & \term{|_M_| ; i -> |_M_|^i} 
  & \term{|_i_|}                 & \term i 
  & \term{|_ a<x>.M _|}          & \term{a<x>.|_M_|} 
  & \term{|_ M^i _|}             & \term{|_M_|^i}
\end{array}
\]
A \emph{parallel reduction step} $\term{M}\rwp_\dup\term{N}$ takes $\term M$ to the marked reduct $\term N = \term{|_M_|}$ for some marking of $\term M$. The \emph{complete development} $\term{\llfloor M\rrfloor}$ of $\term M$ is the marked reduct of marking every duplicating redex in $\term M$.
\end{definition}

Parallel reduction is in-between single-step and multi-step reduction:

\begin{lemma}
\label{lem:dup-complete}
$(\rw_\dup)\subset(\rwp_\dup)\subset(\rws_\dup)$.
\end{lemma}

\begin{proof}
By marking the redex reduced in $\term M\rw_\dup\term N$, respectively by induction on $\term{|_-_|}$.
\end{proof}

Then duplicating reduction and parallel reduction are equivalent, $(\rws_\dup)=(\rwps_\dup)$. Next, a parallel step may be completed to a complete development by another parallel step, reducing the remaining redexes.

\begin{lemma}
\label{lem:dup-par}
If $\term M\rwp_\dup\term N$ then $\term N\rwp_\dup\term{\llfloor M \rrfloor}$. 
\end{lemma}

\begin{proof}
A marking on $\term N$ such that $\term{|_N_|}=\term{\llfloor M\rrfloor}$ will be given by induction on the marked term $\term M$. The two non-trivial cases are when $\term M$ is an unmarked redex:
\begin{align*}
	\term{|_ [M_1]a.a<x>.M_2 _|} &= \term{[|_M_1_|]a.a<x>.|_M_2_|}
\\  \term{|_ M_1^i_|}            &= \term{|_ M_1 _|^i} 			  
\end{align*}
Let $\term{N_k}=\term{|_M_k_|}$ for $k=1,2$ and mark the above redexes in $\term N$. By induction, $\term{|_N_k_|}=\term{\llfloor M_k\rrfloor}$. Then for $\term{|_N_|}$ and $\term{\llfloor M\rrfloor}$ the cases are completed by:
\[
\begin{array}{@{}r@{{}={}}c@{{}={}}l@{}}
        \term{|_ __{[N_1]a.a<x>}.N_2 _|} & \term{\{|_N_1_|/x\}|_N_2_|}	   & \term{\llfloor [M_1]a.a<x>.M_2 \rrfloor}
\\[3pt] \term{|_ __{N_1^i}_|}            & \term{|_N_1_| ; i -> |_N_1_|^i} & \term{\llfloor M_1^i \rrfloor}
\end{array}
\]
\end{proof}

It follows that parallel reduction is diamond, and hence duplicating reduction is confluent. To prove confluence for the full reduction relation $\rw$, terms will be reduced to their affine normal form. First, it is shown how affine reduction commutes with parallel duplicating reduction.

\begin{lemma}
\label{lem:aff-dup}
If $\term P \lwn[\aff] \term M \rwp_\dup \term N$ then $\term P \rwp_\dup \term Q \lws[\aff] \term N$ for some $\term Q$.
\[
\begin{tikzpicture}[x = 30pt , y = -30pt]
	\node (M) at (0,0) {$\term M$};
	\node (N) at (1,0) {$\term N$};
	\node (P) at (0,1) {$\term P$};
	\node (Q) at (1,1) {$\term Q$};
	\draw[->,rwp] (M)--node[above]{$\scriptstyle\dup$} (N);
	\draw[->,rwn] (M)--node[left]{$\scriptstyle\aff$} (P);
	\draw[->,rwp,dashed] (P)--node[below]{$\scriptstyle\dup$} (Q);
	\draw[->,rws,dashed] (N)--node[right]{$\scriptstyle\aff$} (Q);
\end{tikzpicture}
\]
\end{lemma}

\begin{proof}
Let $\term M$ be marked such that $\term{|_M_|}=\term N$. For every affine reduction step $\term M\rw_\aff\term P$ except \emph{prefix (push)} on a marked redex, carrying over the marking from $\term M$ to $\term P$ gives $\term{|_M_|}\rws_\aff\term{|_P_|}$. The remaining case is as follows, where the redex becomes separated so that it cannot be marked.
\[
  \term{__{[N]a.a<x>}.M ; i-> P} \rw_\aff \term{[N]a.(a<x>.M ; i-> P)}
\]
To address this, by confluence and termination the reduction $\term{M}\rwn_\aff\term{P}$ may be arranged so that a \emph{prefix (push)}-step on a marked redex is immediately followed by the \emph{prefix (pop)}-step that restores the redex: 
\[
  \term{[N]a.(a<x>.M ; i-> P)} \rw_\aff \term{__{[N]a.a<x>}.(M ; i-> P)}
\]
This give the required commutation (since $x$ is not free in $\term P$):
\[
\left.\begin{aligned}
   \term{|_ __{[N]a.a<x>}.M ; i-> P_|}
\\ \term{|_ __{[N]a.a<x>}.(M ; i-> P)_|}
\end{aligned}\right\}
~=~
\term{\{|_N_|/x\}|_M_| ; i-> |_P_|}
\]
By induction on $\term M\rwn_\aff\term P$ it follows that $\term{|_M_|}\rws_\aff\term{|_P_|}$.
\end{proof}

Next, define \emph{complete reduction} $(\rwp)=(\rwp_\dup)\cdot(\rwn_\aff)$ as a parallel duplicating step followed by affine normalization.

\begin{lemma}
\label{lem:complete-diamond}
Complete reduction is diamond.
\end{lemma}

\begin{proof}
By the following diagram, where the top left triangles are by Lemma~\ref{lem:dup-complete}, the top right and bottom left squares are by Lemma~\ref{lem:aff-dup}, and the bottom right square is by confluence and termination of affine reduction (Lemma~\ref{lem:aff}).
\[
\begin{tikzpicture}[x = 30pt , y = -30pt]
	\node (M)  at (0,0) {$\term M$} ;
	\node (MN) at (1,0) {$\term\cdot$} ;
	\node (N)  at (2,0) {$\term N$} ;
	\node (MP) at (0,1) {$\term\cdot$} ;
	\node (P)  at (0,2) {$\term P$} ;
	\node (NQ) at (2,1) {$\term\cdot$} ;
	\node (PQ) at (1,2) {$\term\cdot$} ;
	\node (C)  at (1,1) {$\term\cdot$} ;
	\node (Q)  at (2,2) {$\term Q$} ;
	\draw[->,rwp]  (M)--node[above]{$\scriptstyle\dup$} (MN);
	\draw[->,rwp]  (M)--node[left] {$\scriptstyle\dup$} (MP);
	\draw[->,rwp]  (M)--node[above right]{$\scriptstyle\dup$} (C);
	\draw[->,rwp] (MN)--node[right]{$\scriptstyle\dup$} (C);
	\draw[->,rwp] (MP)--node[below]{$\scriptstyle\dup$} (C);
	\draw[->,rwp]  (N)--node[right]{$\scriptstyle\dup$} (NQ);
	\draw[->,rwp]  (P)--node[below]{$\scriptstyle\dup$} (PQ);
	\draw[->,rwn] (MN)--node[above]{$\scriptstyle\aff$} (N);
	\draw[->,rwn] (MP)--node[left] {$\scriptstyle\aff$} (P);
	\draw[->,rwn] (NQ)--node[right]{$\scriptstyle\aff$} (Q);
	\draw[->,rwn] (PQ)--node[below]{$\scriptstyle\aff$} (Q);
	\draw[->,rws]  (C)--node[below]{$\scriptstyle\aff$} (NQ);
	\draw[->,rws]  (C)--node[right]{$\scriptstyle\aff$} (PQ);
\end{tikzpicture}
\]
\end{proof}

To connect reduction to complete reduction, the following lemma will show that the image of a reduction step under affine normalization is a complete step.

\begin{lemma}
\label{lem:aff-complete}
If $\term M\rw\term N$ then $\term{M_\aff}\rwp\term{N_\aff}$ where $\term{M_\aff}$ and $\term{N_\aff}$ are the affine normal forms of $\term M$ and $\term N$.
\end{lemma}

\begin{proof}
The case of an affine step $\term M\rw_\aff\term N$ is immediate since $\term{M_\aff}=\term{N_\aff}$ by confluence and termination of affine reduction (Lemma~\ref{lem:aff}). In the case of a duplicating step $\term M\rw_\dup\term N$, Lemma~\ref{lem:dup-par} gives a parallel step $\term M\rwp_\dup\term N$, for which Lemma~\ref{lem:aff-dup} gives reductions $\term{M_\aff}\rwp_\dup\term{P}\lws[\aff]\term{N}$ for some term $\term P$. Then $\term{P}\rwn_\aff\term{N_\aff}$ again by Lemma~\ref{lem:aff}, giving the required reduction $\term{M_\aff}\rwp_\dup\term{P}\rwn_\aff\term{N_\aff}$.
\[
\begin{tikzpicture}[x = 30pt , y = -30pt, baseline=-36pt]
	\node (M)  at (0,0) {$\term M$};
	\node (N)  at (1,0) {$\term N$};
	\node (Ma) at (0,1) {$\term{M_\aff}$};
	\node (Na) at (1,1) {$\term{N_\aff}$};	
	\node      at (.5,1) {$=$};
	\draw[->,rw]  (M)--node[above]{$\scriptstyle\aff$} (N);
	\draw[->,rwn] (M)--node[left] {$\scriptstyle\aff$} (Ma);
	\draw[->,rwn] (N)--node[right]{$\scriptstyle\aff$} (Na);
\end{tikzpicture}
\qquad
\begin{tikzpicture}[x = 30pt , y = -30pt, baseline=-36pt]
	\node (M)  at (0,0) {$\term M$};
	\node (N)  at (1,0) {$\term N$};
	\node (Ma) at (0,1) {$\term{M_\aff}$};
	\node (P)  at (1,1) {$\term P$};
	\node (Na) at (2,1) {$\term{N_\aff}$};
	\draw[->,rw]   (M)--node[above]{$\scriptstyle\dup$} (N);
	\draw[->,rwn]  (M)--node[left] {$\scriptstyle\aff$} (Ma);
	\draw[->,rws]  (N)--node[right]{$\scriptstyle\aff$} (P);
	\draw[->,rwn]  (N)--node[above]{$\scriptstyle\aff$} (Na);
	\draw[->,rwp] (Ma)--node[below]{$\scriptstyle\dup$} (P);
	\draw[->,rwn]  (P)--node[below]{$\scriptstyle\aff$} (Na);
\end{tikzpicture}
\]
\end{proof}

The confluence proof puts everything together: affine normalization maps reduction to complete reduction, which is confluent, and which is included in reduction.

\begin{theorem}
\label{thm:confluence}
Reduction $\rw$ is confluent.
\end{theorem}

\begin{proof}
Let $\term{P}\lws\term{M}\rws\term{N}$. By Lemma~\ref{lem:aff-complete} there are complete reductions $\term{P_\aff}\lwps\term{M_\aff}\rwps\term{N_\aff}$ where $\term{M_\aff}$, $\term{N_\aff}$, and $\term{P_\aff}$ are the affine normal forms of respectively $\term M$, $\term N$, and $\term P$. The diamond property for complete reduction (Lemma~\ref{lem:complete-diamond}) gives reductions $\term{P_\aff}\rwps\term{Q}\lwps\term{N_\aff}$ for some $\term Q$. Since a parallel step $\rwp_\dup$ corresponds to a duplicating reduction $\rws_\dup$ (Lemma~\ref{lem:dup-complete}), a complete step $\rwp$ corresponds to a reduction $\rws$, which gives the desired converging reductions $\term{P}\rwn_\aff\term{P_\aff}\rws\term{Q}\lws\term{N_\aff}\lwn[\aff]\term{N}$.
\[
\begin{tikzpicture}[x = 30pt , y = -30pt]
	\node (M)  at (0,0) {$\term M$};
	\node (N)  at (2,0) {$\term N$};
	\node (P)  at (0,2) {$\term P$};
	\node (Ma) at (1,1) {$\term{M_\aff}$};
	\node (Na) at (2,1) {$\term{N_\aff}$};	
	\node (Pa) at (1,2) {$\term{P_\aff}$};
	\node (Q)  at (2,2) {$\term Q$};
	\draw[->,rws]  (M)--(N);
	\draw[->,rws]  (M)--(P);
	\draw[->,rwn]  (M)--node[above]{$\scriptstyle\aff$} (Ma);
	\draw[->,rwn]  (N)--node[right]{$\scriptstyle\aff$} (Na);
	\draw[->,rwn]  (P)--node[above]{$\scriptstyle\aff$} (Pa);
	\draw[->,rwps] (Ma)--(Na);
	\draw[->,rwps] (Ma)--(Pa);
	\draw[->,rwps] (Na)--(Q);
	\draw[->,rwps] (Pa)--(Q);
\end{tikzpicture}
\]
\end{proof}


\section{Types}

This section formally introduces the simply typed FMC with control. Types are stratified into four layers: types for terms, vectors of types for stacks, location-indexed families of vectors for memories, and choice-indexed families of memory types as return types. The formal definitions use indexed families directly, while the notation of the informal introduction is given as operations on types or as syntactic sugar.
\[
\begin{array}{lr@{}l@{\qquad\quad}l}
    	\text{Types:}        & \type{r,s,t}  & ~\coloneqq~ \type{!!s => !!tI}	
\\[2pt] \text{Stack types:}  & \type{ !t}    & ~\coloneqq~ \type{t_1..t_n}  
\end{array}
\qquad
\begin{array}{lr@{}l@{\qquad\quad}l}
    	\text{Memory types:} & \type{!!t}    & ~\coloneqq~ \{\type{!t_a}\mid a\in A\} 
\\[2pt] \text{Sum types:}    & \type{!!tI }  & ~\coloneqq~ \{\type{!!t_i}\mid i\in I\}
\end{array}
\]
The base cases are given by empty stack types and memory types, both written $\type\e$, and the empty sum type, written $\type0$. Vectors are composed by juxtaposition $\type{!s!t}$, lifted to families point-wise: $\type{!!s\,!!t}=\{\type{!s_a\,!t_a}\mid a\in A\}$. The singleton family holding $\type{!t}$ at location $a$ and empty elsewhere is written $\type{a(!t)}$. This retrieves the notation $\type{a_1(t_1)..a_n(t_n)}$ for memory types in the introductory sections, and generalizes it to allow type vectors $\type{a_1(!t_1)..a_n(!t_n)}$. 

Sum types are combined by $\type{!!sI+!!tJ}$ where $I\cap J=\varnothing$. To retrieve the notation from the informal development, the notation $\type{!!t.i}$ indicates the singleton family over $\{i\}$ containing only the memory type $\type{!!t}$ for the choice $i$; formally, $\type{!!t.i}$ is $\type{!!t_{\{i\}}}$ where $\type{!!t_i}=\type{!!t}$. That is, $\type{!!t_i}$ is a \emph{memory type} as a member of a family, and $\type{!!t.i}$ is a \emph{sum type} that is a singleton family. Sum types may then be written $\type{!!t_1.i_1+..+!!t_n.i_n}$ as before. Finally, $\type{!!t_{I\setminus i}}$ (respectively $\type{!!t_{I\setminus J}}$) denotes the family $\type{!!tI}$ minus the element $\type{!!t_i}$ (respectively the elements $\type{!!t_j}$ for $j\in J$), if present.

Stack types follow the order of terms on the stack, $\term{MNP:rst}$ for $\term{M:r}$, $\term{N:s}$, and $\term{P:t}$. Since stacks are last-in-first-out, identity terms are of the form $\term{<x>.<y>.<z>.[z].[y].[x].*}$, with the order of pops reversed relative to a given input stack. The convention is then that stack types on the left of an implication are presented in reverse order, i.e.\ the type for this term would be $\type{\lambda(tsr)=>\lambda(rst).\star}$, written $\type{\lambda(!t)=>\lambda(!t).\star}$ for the stack type $\type{!t}=\type{rst}$.

A \emph{context} $\Gamma$ is a finite function from variables to types, written as a sequence $\term{x_1:t_1,,x_n:t_n}$. A \emph{typing judgement} $\term{G |- M:t}$ assigns the type $\type t$ to the term $\term M$ in the context $\Gamma$.

\begin{definition}
The \emph{simply-typed FMC with control} is given by the typing rules in Figure~\ref{fig:types}.
\end{definition}

\begin{figure}
\[
\begin{array}{@{}c@{}}
	   \infer[\sr x]{\term{G , x: t |- x: t}}{}
\qquad \infer[\sr c]{\term{G |- i : \e => \e.i }}{}
\\ \\
\begin{array}{@{}c@{\qquad}c@{\qquad}c@{}}
	  \infer[\sr a]{\term{G |- [N]a.M : ??s => !!tI}}{
		\term{G |- N : r}
	  & \term{G |- M : a(r)\,??s => !!tI}
	  }
&	  \infer[\sr l]{\term{G |- a<x>.M : a(r)\,??s => !!tI}}{\term{G , x:r |- M : ??s => !!tI}}
&     \infer[\sr e]{\term{G |- M: ??r\,??s => (!!s\,!!t)_I}}{\term{G |- M: ??r => !!tI}}
\\ \\
	  \infer[\sr ;]{\term{G |- M;i->N : ??r => !!tI}}{
		\term{G |- M: ??r => !!t_{I\setminus i} + !!s.i}
	  & \term{G |- N: ??s => !!tI}
	  }	  
&     \infer[\sr o]{\term{G |- M^i : ??s => !!tI}}{\term{G |- M: ??s => !!tI + !!s.i}}
&     \infer[\sr i]{\term{G |- M: ??r => !!sI + !!tJ}}{\term{G |- M: ??r => !!sI}}
\end{array}
\end{array}
\]
\caption{The simply-typed FMC with control}
\label{fig:types}
\end{figure}

A few notes on the typing rules. The rules $\SR a$ and $\SR l$ are the equivalent of the rules for application and abstraction for the simply-typed $\lambda$-calculus, since the arrow type $\typ{s->t}$ is interpreted as introducing an additional input type $\typ s$ to the input type vector of $\typ t$. The rule $\SR e$ for \emph{(stack) expansion} extends the input and output memory types of a term by $\type{!!s}$, on every output branch, reflecting the principle of stack calculi that terms may operate on arbitrarily large stacks, returning any additional part untouched. The rule $\SR;$ for $\term{M;i->N}$ requires the output of $\term M$ on choice $i$ to match the \emph{input} of $\term N$, and on any other choice $I\setminus i$ to match the \emph{output} of $\term N$; the output type $\type{!!tI}$ of $\term N$ may or may not have a component $\type{!!t_i}$.

Observe that the typing rules as given are not inductive on terms, due to the rules $\SR e$ (\emph{expansion}) and $\SR i$ (\emph{inclusion}). This gives a simpler presentation and reduces repetition in proofs. Both rules can however be permuted up past the other rules, which means they may instead be integrated into the \emph{variable} and \emph{choice} rules; or they may permuted down to be integrated into \emph{push} and \emph{case} rules.

Next, the basic properties of preservation of types under substitution and reduction are given.

\begin{lemma}[Subject substitution]
\label{lem:substitution}
If $\term{G |- N:s}$ and $\term{G , x:s |- M:t}$ then $\term{G |- \{N/x\}M : t}$.
\end{lemma}

\begin{proof}
By induction on the typing derivation for $\term M$.
\end{proof}

\begin{proposition}[Subject reduction]
If $\term{G |- M:t}$ and $\term M\rw\term N$ then $\term{G |- N:t}$.
\end{proposition}

\begin{proof}
By induction on the typing derivation for $\term M$, with top-level reduction steps as base cases, and using the subject substitution lemma (Lemma~\ref{lem:substitution}) in the case of a \emph{beta} step.
\end{proof}

To demonstrate that machine evaluation preserves types, the type system is extended to stacks, memories, and machine states by the rules in Figure~\ref{fig:extra-types}. Continuation stacks will have types $\type{!!sI => !!tJ}$, extending the grammar of types. 

\begin{proposition}[Machine evaluation preserves types]
For a typed state $\term{|- +{S_A}MK : \e => !!tI}$, if
\[
	\step {S_A}MK {T_A}NL
\]
then $\term{|- +{T_A}NL : \e => !!tI}$.
\end{proposition}

\begin{proof}
By inspection of the machine transitions.
\end{proof}

\begin{proposition}[Machine progress]
\label{prop:progress}
A typed state is either final or has a machine step.
\end{proposition}

\begin{proof}
Recall that states are either final, have a machine step, or are failure states, which are those with a free variable $\term x$ as term or with an abstraction $\term{a<x>.M}$ and an empty stack $\e_a$ in the memory. Both cases are ruled out by the type system.
\end{proof}

\begin{figure}
\[
\begin{array}{c@{\qquad}c}
  \infer[\sr s]{\term{|- \e : \e}}{}
& \infer[\sr k]{\term{|- \e : !!tI => !!tI}}{}
\\ \\  
  \infer[\sr S]{\term{|- `{\black S}\,M : !s\,t}}{\term{|- `{\black S}: !s} && \term{|- M:t}}
& \infer[\sr K]{\term{|- (i -> M)\,`{\black K} : !!r.i + !!s_{I\setminus i} => !!tJ}}{\term{|- M : !!r => !!sI + !!t_{J\setminus I}} && \term{|- `{\black K} : !!sI => !!tJ}}
\\ \\
  \infer[\sr M]{\term{|- `{\black S_A} : \{!s_a\mid a\in A\}}}{\{\term{|- `{\black S_a} : !s_a}\}_{a\in A}}
& \infer[\sr +]{\term{|- +{S_A}MK : \e => !!tJ}}{\term{|- S_A : !!r} & \term{|- M : !!r => !!sI + !!t_{J\setminus I}} & \term{|- K : !!sI => !!tJ}}
\end{array}
\]
\caption{Extended types for stacks, memories, and states}
\label{fig:extra-types}
\end{figure}

In the FMC, unlike the $\lambda$-calculus, all types are inhabited. A \emph{zero} term for a type $\type{t}=\type{!!s=>!!tI}$ will evaluate on the machine by discarding a memory of $\type{!!s}$, and for some $i\in I$ returning a memory of zero-terms of type $\type{!!t_i}$. Zero terms are formally defined as follows: if $I$ is non-empty, select $i\in I$ and define:
\[
	\term{0_{\typ t}} = \term{<\__{\typ{!!s}}>.[0_{\typ{!!t_i}}].i}
\]
where $\term{<\__{\typ{!!s}}>}$ is a sequence of non-binding abstractions $\term{a_1<\_>\dots a_n<\_>}$ matching $\type{!!s}=\type{a_1(s_1)..a_n(s_n)}$, and $\term{[0_{\typ{!!t}}]}$ is a sequence of applications $\term{[0_{\typ{t_1}}]a_1\dots[0_{\typ{t_n}}]a_n}$ matching $\type{!!t_i}=\type{a_1(t_1)..a_n(t_n)}$. For zero terms where $I$ is empty, where $\type{t}=\type{!!s=>0}$, let $\type{t'}=\type{!!s=>!!s.\star}$ and define the zero term as the loop $\term{0_{\typ t}}=\term{(0_{\typ{t'}})^\star}$. Note that for a type $\type{!!s=>!!tI}$ and a given $i\in I$, zero-terms are equivalent by the permutations $\term{a<x>.b<y>.M}\sim\term{b<y>.a<x>.M}$ and $\term{[P]a.[N]b.M}\sim\term{[N]b.[P]a.M}$ where $a\neq b$.

\begin{proposition}[Type inhabitation]
\label{prop:inhabitation}
Every type $\type t$ is inhabited by a zero term $\term{|- 0_{\typ t}:t}$.
\end{proposition}

\begin{proof}
By induction on the type $\type t$.
\end{proof}


\section{Machine termination}
\label{sec:termination}

For terms without loops, types guarantee termination of the machine. By this, the following is meant.

\begin{definition}
A term $\term M$ is \emph{terminating} if $\evl{S_A}M{T_A}i$ for some memories $S_A$, $T_A$ and choice $i$.
\end{definition}

The typed notion, proved below, is stronger. For a term $\term{M:!!s=>!!tI}$, for \emph{any} memory $S_A:\type{!!s}$ of the correct type there are $i\in I$ and $T_A:\type{!!t_i}$ such that $\evl{S_A}M{T_A}i$. Type inhabitation guarantees that a suitable input memory $S_A:\type{!!s}$ exists, so that the typed notion implies the untyped definition above. 

This will be proved using the standard Tait reducibility technique~\cite{Tait-1967}. Each type $\type t$ is associated with a set $\RUN{t}$ of terminating terms, here called the \emph{runnable} terms by analogy to the usual \emph{reducible} terms of strong normalization arguments. It is then shown by induction on typing derivations that every typed term is runnable, and hence terminating.

\begin{definition}
The set $\RUN{t}$ of \emph{runnable terms} for a type $\type t$ is defined as the set of closed terms
\[
\RUN{!!s => !!tI} = \{ \term M ~\mid~ \forall S_A\in\RUN{!!s}.~\exists i\in I.~\exists T_A\in\RUN{!!t_i}.~\evl{S_A}M{T_A}i~\}
\]
where the runnable sets for memory types $\type{!!t}$ and stack types $\type{!t}$ are as follows.
\[
\begin{aligned}
    \RUN{\{!t_a\mid a\in A\}} &= \{ S_A \mid \forall a\in A.~S_a\in\RUN{!t_a} \}
&&& \RUN{t_1..t_n}            &= \{ \e\,\term{M_1}\dots \term{M_n} \mid \term{M_i}\in\RUN{t_i} \}
\end{aligned}
\]
\end{definition}

To work with open terms, a \emph{substitution map} $\subs$ is a finite function from variables to terms, applied to terms as $\term{\subs M}$ as a simultaneous substitution for the variables in its domain. Denote by $\subs\{\term M/x\}$ the map that assigns $\term M$ to $x$ and is as $\subs$ for other variables. Then $\SEQ{G}$ associates a context $\Gamma$ with the substitution maps over runnable terms of the types in $\Gamma$:
\[
	\SEQ{x_1:t_1,,x_n:t_n} = \{\subs\mid \term{\subs x_i}\in\RUN{t_i}\}
\]

The next lemma will show that any loop-free, typed term is runnable. It assumes return types to be non-empty, i.e.\ types are not of the form $\type{!!s=>0}$. The proof is then a direct induction on typing derivations.

\begin{lemma}[Typed terms are runnable]
\label{lem:runnable}
If $\term{G |- M:t}$ then $\term{\subs M}\in\RUN{t}$ for any $\subs\in\SEQ{G}$.
\end{lemma}


\begin{proof}
By induction on the typing derivation for $\term{G |- M:t}$.

\begin{itemize}
	\item\emph{Variable:}
\[
\infer[\sr x]{\term{G , x: t |- x: t}}{}
\]
For any $\subs\in\SEQ{G , x:t}$ by definition $\term{\subs x}\in\RUN{t}$.

	\item\emph{Push:}
\[
\infer[\sr a]{\term{G |- [N]a.M : ??s => !!tI}}{
		\term{G |- N : r}
	  & \term{G |- M : a(r)\,??s => !!tI}
	  }
\]
Let $\subs\in\SEQ{G}$ and $S_A\in\RUN{!!s}$. By the inductive hypothesis, $\term{\subs N}\in\RUN{r}$, so that $S_A\msep a(\term{\subs N})\in\RUN{!!s\,a(r)}$. Again by the inductive hypothesis, $\term{\subs M}\in\RUN{a(r)\,??s=>!!tI}$, which means it evaluates as $\eval{S_A\msep a(\term{\subs N})}{\subs M}{T_A}i$ for some $i\in I$ and $T_A\in\RUN{!!t_i}$. Then $\eval{S_A}{[\subs N]a.\subs M}{T_A}i$ by the definition of $(\evalarrow)$, which gives $\term{\subs([N]a.M)}=\term{[\subs N]a.\subs M}\in\RUN{??s => !!tI}$.

	\item\emph{Pop:}
\[
\infer[\sr l]{\term{G |- a<x>.M : a(r)\,??s => !!tJ}}{\term{G , x:r |- M : ??s => !!tI}}
\]
Let $\subs\in\SEQ{G}$ and $S_A\msep a(\term N)\in\RUN{!!s\,a(r)}$, and assume by $\alpha$-equivalence that $x$ is not in the domain of $\Gamma$. Then $S_A\in\RUN{!!s}$ and $\subs\{\term{N}/x\}\in\SEQ{G,x:r}$. Since $\term N$ is closed, $\term{(\subs\{N/x\})M}=\term{\subs(\{N/x\}M)}$. By the inductive hypothesis, $\term{\subs(\{N/x\}M)}\in\RUN{??s => !!tI}$ so that $\eval{S_A}{\subs(\{N/x\}M)}{T_A}i$ for some $i\in I$ and $T_A\in\RUN{!!t_i}$. By the definition of $(\evalarrow)$ it follows that $\eval{S_A\msep a(\term N)}{\subs(a<x>.M)}{T_A}i$ and hence $\term{\subs(a<x>.M)}\in\RUN{a(r)\,??s => !!tI}$.

	\item\emph{Choice:}
\[
\infer[\sr c]{\term{G |- i : \e => \e.i }}{}
\]
Note that the empty memory $\e\in\RUN{\e}$ is the only inhabitant of $\RUN{\e}$, and that $\term{\subs i}=\term i$ for any substitution map $\subs$. Since $\eval\e i\e i$ it follows that $\term i\in\RUN{\e => \e.i}$.

	\item\emph{Case:}
\[
\infer[\sr ;]{\term{G |- N;i->M : ??r => !!tI}}{
		\term{G |- N: ??r => !!t_{I\setminus i} + !!s.i}
	  &&\term{G |- M: ??s => !!tI}
	  }
\]
Let $\subs\in\SEQ{G}$ and $R_A\in\RUN{!!r}$. The inductive hypothesis gives:
\[
		\term{\subs N}\in\RUN{??r => !!t_{I\setminus i}+!!s.i} \qquad\text{and}\qquad  \term{\subs M}\in\RUN{??s => !!tI}
\]
For the former, there are two cases:
\[
	\eval{R_A}{\subs N}{S_A}i  \qquad \text{or} \qquad \eval{R_A}{\subs N}{T_A}j \quad\text{where}\quad j\in I,~i\neq j
\]
In the first case, $S_A\in\RUN{!!s}$ and hence $\eval{S_A}{\subs M}{T_A}k$ for some $k\in I$ and $T_A\in\RUN{!!t_k}$. The definition of $(\evalarrow)$ gives the evaluation below left. For the second case there is the evaluation below right. It follows that $\term{\subs(M;i->N)}\in\RUN{??r => !!tI}$.
\[
	\bigstep
	  { \eval {R_A}  M       {S_A} i \quad
	    \eval {S_A}  N       {T_A} k }
	  { \eval {R_A} {M;i->N} {T_A} k }
\qquad
\qquad
	\bigstep
	  { \eval {R_A}  M       {T_A} j }
	  { \eval {R_A} {M;i->N} {T_A} j }
\]

	\item\emph{Expansion}
\[
\infer[\sr e]{\term{G |- M: ??r\,??s => (!!s\,!!t)_I}}{\term{G |- M: ??r => !!tI}}
\]
By induction on $(\evalarrow)$, if $\evl{R_A}M{T_A}j$ then $\evl{S_AR_A}M{S_AT_A}i$. The case is then immediate.

	\item\emph{Inclusion}
\[
\infer[\sr i]{\term{G |- M: ??r => !!sI + !!tJ}}{\term{G |- M: ??r => !!sI}}
\]
Immediate since a return memory $S_A\in\RUN{!!s_i}$ for a choice $i\in I$ is also one for $i\in I\cup J$.
\end{itemize}
\end{proof}

The theorem is then as follows.

\begin{theorem}[Machine termination]
\label{thm:termination}
Any loop-free, typed term $\term{|- M:!!s=>!!tI}$ with a zero-free type is terminating.
\end{theorem}

\begin{proof}
By type inhabitation there is a loop-free memory $S_A:\type{!!s}$. By Lemma~\ref{lem:runnable}, $\term{M}\in\RUN{!!s=>!!tI}$ and $S_A\in\RUN{!!s}$. By the definition of runnable terms, $\evl{S_A}M{T_A}i$ for some $i$ and $T_A$.
\end{proof}


\section{Strong normalization}
\label{sec:normalization}

This section will prove that typed, loop-free terms are strongly normalizing. The proof is centered around the idea that beta-reduction shortens a run of the machine by eliminating consecutive push- and pop-transitions. This gives a proof in two stages. First, a reducibility argument along the lines of that for machine termination establishes that typed terms have a finite evaluation on the machine, for which the \emph{measured evaluation} relation defined below counts the number of pop-transitions. Second, it is shown that \emph{beta}-reduction strictly reduces this measure, while affine reduction does not increase it. Since affine reduction is inherently normalizing, this establishes typed strong normalization.

The first stage needs to overcome the obstacle that terms that are discarded without evaluation must somehow be measured as well. This is familiar from many strong normalization proofs, such as those deriving strong normalization from weak normalization~\cite{Nederpelt-1973,Klop-1980} and those using perpetual reduction strategies~\cite{Bergstra-Klop-1982,Raamsdonk-Severi-Sorensen-Xi-1999}, and can be distinguished in the combinatoric arguments of many others. Here, it is addressed by requiring machine evaluation of all subterms: for an application $\term{[N]a.M}$ the argument $\term N$ is evaluated as well as pushed to the stack, and for a case $\term{M;i->N}$ the continuation $\term N$ is evaluated even when $\term M$ terminates with $j\neq i$. Evaluating such terms requires an input stack, which for typed terms is guaranteed by type inhabitation (Proposition~\ref{prop:inhabitation}).

The proof thus comprises a \emph{logical} part, using abstract reducibility, and a \emph{combinatorial} part, measuring abstraction steps in machine evaluation and demonstrating that the measure reduces under \emph{beta}-reduction. This may be viewed as a decomposition of the proof for the previous iteration of the FMC~\cite{Barrett-Heijltjes-McCusker-2023}, which followed the style of Gandy's proof~\cite{Gandy-1980,Fuhs-Kop-2012}, computing the same measure directly from typing derivations. The key to this decomposition is type inhabitation: Gandy's proof, which maps terms onto domains of monotone functionals over integers, relies on the existence of minimal elements, essentially higher-order versions of the constant zero function, to provide input to any function. Here, because all types are inhabited, zero-terms $\term{0}_\type{t}$ may fulfil that r\^ole.

This section will again assume loop-free terms (i.e.\ not containing the loop construct $\term{M^i}$) and non-empty sum types (i.e.\ $I\neq\varnothing$ for any $\type{!!tI}$).

\begin{definition}
\emph{Measured evaluation} $\evl[n]{S_A}M{T_A}i$ is defined inductively as follows, where $i\neq j$.
\[
\bigstep{}{\eval[0]{S_A} i {S_A} i}
\quad
\begin{array}{@{}r@{\quad~~}r@{\quad~~}r@{}}
	\bigstep  
	  { \eval[n]   {R_A}                     N  {U_A} k \quad
	    \eval[m]   {S_A\msep a(\term N)}     M  {T_A} i }
	  { \eval[n+m] {S_A}               {[N]a.M} {T_A} i }
&
	\bigstep
	  { \eval[m]   {R_A}  M       {S_A} i \quad
	    \eval[n]   {S_A}  N       {T_A} k }
	  { \eval[m+n] {R_A} {M;i->N} {T_A} k }
\\ \\[-10pt]
	\bigstep
	  { \eval[n]   {S_A}               {\{N/x\}M}  {T_A} i }
	  { \eval[n+1] {S_A\msep a(\term N)} {a<x>.M}  {T_A} i }
&
	\bigstep
	  { \eval[m]   {R_A}  M       {S_A} j \quad
	    \eval[n]   {T_A}  N       {U_A} k }
	  { \eval[m+n] {R_A} {M;i->N} {S_A} j }
\end{array}
\]
\end{definition}

The reducibility sets $\EVL-$ interpret types as a guarantee that a measured evaluation exists.

\begin{definition}
The set $\EVL{t}$ for a type $\type t$ is defined as the set of closed terms
\[
\EVL{!!s => !!tI} = \{ \term M ~\mid~ \forall S_A\in\EVL{!!s}.~\exists m\in\mathbb N.~\exists i\in I.~\exists T_A\in\EVL{!!t_i}.~\evl[m]{S_A}M{T_A}i~\}
\]
where for memory types $\type{!!t}$ and stack types $\type{!t}$:
\[
\begin{aligned}
    \EVL{\{!t_a\mid a\in A\}} &= \{ S_A \mid \forall a\in A.~S_a\in\EVL{!t_a} \}
&&& \EVL{t_1..t_n}            &= \{ \e\,\term{M_1}\dots \term{M_n} \mid \term{M_i}\in\EVL{t_i} \}
\end{aligned}
\]
For contexts, $\SEQ{G}$ is a set of substitution maps $\subs$:
\[
	\EVS{x_1:t_1,,x_n:t_n} = \{\subs\mid \term{\subs(x_i)}\in\EVL{t_i}\}
\]
\end{definition}

Evaluation sets are inhabited at least by zero-terms.

\begin{lemma}
\label{lem:zero-eval}
A zero-term $\term{0_{\typ t}}$ is in $\EVL{t}$.
\end{lemma}

\begin{proof}
For a type $\type s$, let $\type{\src(s)}$ and $\type{\tgt(s)}$ be its source (memory) type and target (sum) type, so that $\type s=\type{\src(s)=>\tgt(s)}$. By induction on $\type t$ it will be shown for all $S_A$ of the dimensions of $\type{\src(t)}$ that $\eval[n]{S_A}{0_{\typ t}}{0_{\type{\tgt(t)_i}}}i$ for some $n$ and $i$. This immediately implies $\term{0_{\typ t}}\in\EVL{t}$, since any $S_A\in\EVL{\src(t)}$ will be of the right dimensions, and inductively $0_{\type{\tgt(t)_i}}\in\EVL{\tgt(t)_i}$ for the return memory. Let $\term{0_{\typ t}}$ be as follows.
\[
	\term{0_{\type t}} = \term{<\_\kern1pt_{\typ{\src(t)}}>.[0_{\typ{\tgt(t)_i}}].i}
\]
By induction, for every $\type{s}$ in $\type{\tgt(t)_i}$ there is the following evaluation for some $j_\type{s}$ and $m_{\type s}$, since $0_{\type{\src(s)}}$ is of the required dimensions.
\[
	\eval[m_{\type s}]{0_{\type{\src(s)}}}{0_{\typ s}}{0_{\type{\tgt(s)_j}}}{j_{\typ s}}
\]
Then for any memory $S_A$ of the dimensions of $\type{\src(t)}$ there is the following evaluation, where the double lines indicate multiple evaluation rules, $m$ is the sum over the measures $m_{\type s}$ for each $\type{s}$ in $\type{\tgt(t)_i}$, and $n$ is the total size of $\type{\src(t)}$.
\[
\begin{array}{@{}r@{}}
	  \{\eval[m_{\type s}]{0_{\type{\src(s)}}}{0_{\typ s}}{0_{\type{\tgt(s)_j}}}{j_{\typ s}}\}_{\type{s}\in\type{\tgt(t)_i}} \quad
	  \overline{\eval[0]{0_{\type{\tgt(t)_i}}}i{0_{\type{\tgt(t)_i}}}i}
\\ \dline
\begin{array}{@{}r@{}}
	  \eval[m]{\e}{[0_{\typ{\tgt(t)_i}}].i}{0_{\type{\tgt(t)_i}}}i
\\ \dline
	\eval[n+m]{S_A}{<\_\kern1pt_{\typ{\src(t)}}>.[0_{\typ{\tgt(t)_i}}].i}{0_{\type{\tgt(t)_i}}}i
\end{array}
\end{array}
\]
\end{proof}

The main reducibility lemma then shows that types guarantee measured evaluation.

\begin{lemma}
\label{lem:type-eval}
If $\term{G |- M:t}$ and $\subs\in\EVS{G}$ then $\term{\subs M}\in\EVL{t}$.
\end{lemma}

\begin{proof}
By induction on the typing derivation for $\term M$. The proof is similar to that of Lemma~\ref{lem:runnable} that typed terms are runnable.
\end{proof}

Reduction $(\rw)$ reduces the measure given by measured evaluation, by the following lemma.

\begin{lemma}
\label{lem:eval-reduce}
If $\evl[n]{S_A}{M}{T_A}i$ and $S_A\rws S'_A$ and $\term M\rws\term{M'}$, then $\evl[n']{S'_A}{M'}{T'_A}i$ where $T_A\rws T'_A$ and $n\geq n'$. If moreover $\term M\rws\term{M'}$ contains beta-steps, then $n>n'$.
\end{lemma}
%

\begin{proof}
By an outer induction on the size of the derivation for $\evalarrow_n$, and an inner induction on the length of the reduction $\term{M}\rws\term{M'}$.

\begin{itemize}
	\item\emph{Choice:}
\[
\begin{array}{@{}c@{}}\\\hline\eval[0]{S_A} i {S_A} i\end{array}
\]
For $S_A\rws S'_A$ it is immediate that $\eval[0]{S'_A} i {S'_A} i$.
	
	\item\emph{Pop:}
\[
	\bigstep
	  { \eval[n]   {S_A}               {\{N/x\}M}  {T_A} i }
	  { \eval[n+1] {S_A\msep a(\term N)} {a<x>.M}  {T_A} i }
\]
Let $S_A\rws S'A$ and $\term N\rws\term{N'}$, and note that the reduction on $\term{a<x>.M}$ must be of the form $\term{a<x>.M}\rws\term{a<x>.M'}$. Since $\term{\{N/x\}M}\rws\term{\{N'/x\}M'}$ the outer inductive hypothesis gives the required $T'_A$ and $n'$.

	\item\emph{Push:}
\[
	\bigstep  
	  { \eval[n]   {R_A}                     N  {U_A} k \quad
	    \eval[m]   {S_A\msep a(\term N)}     M  {T_A} i }
	  { \eval[n+m] {S_A}               {[N]a.M} {T_A} i }
\]
There are three cases, depending on the first reduction step on $\term{[N]a.M}$: 1) reduction in either subterm $\term M$ or $\term N$, 2) a top level beta-step, or 3) a top level passage step. For case 1), given $S_A\rws S'_A$ and $\term{[N]a.M}\rw\term{[N']a.M'}$ (where $\term{N}=\term{N'}$ or $\term{M}=\term{M'}$), the outer inductive hypothesis for $\term M\rw\term{M'}$ gives the required $T'_A$ and $m'$, or that for $\term N\rw\term{N'}$ gives $n'$. For the remaining reduction from $\term{[N']a.M'}$ the statement follows by the inner inductive hypothesis.

In case 2) $\term M=\term{a<x>.P}$ with a beta-step $\term{[N]a.a<x>.P}\rw\term{\{N/x\}P}$. Evaluation is as follows.
\[
	\bigstep  
	  { \eval[n] {R_A} N  {U_A} k \quad
	    \bigstep{ \eval[m]   {S_A}                {\{N/x\}P} {T_A} i }
	            { \eval[m+1] {S_A\msep a(\term N)}  {a<x>.P} {T_A} i } }
	  { \eval[n+m+1] {S_A} {[N]a.a<x>.P} {T_A} i }
\]
Given $S_A\rws S'_A$ and the remaining reduction $\term{\{N/x\}P}\rws\term{M'}$, the (outer) inductive hypothesis gives the the evaluation $\evl[m']{S'_A}{M'}{T'_A}i$ where $T_A\rws T'_A$ and $m\geq m'$, as required.

In case 3) $\term M=\term{b<x>.P}$ with a passage step $\term{[N]a.b<x>.P}\rw\term{b<x>.[N]a.P}$ where $a\neq b$ and $x\notin\fv(\term N)$. Evaluation for the redex is below left, and for the reduct below right, with a derivation of the same size. The inner inductive hypothesis then gives the required $T'_A$, $m'$, and $n'$.
\[
	\bigstep  
	  { \eval[n] {R_A} N  {U_A} k \quad
	    \bigstep{ \eval[m]   {S_A\msep a(\term N)}                {\{Q/x\}P} {T_A} i }
	            { \eval[m+1] {S_A\msep b(\term Q)\msep a(\term N)}  {b<x>.P} {T_A} i } }
	  { \eval[n+m+1] {S_A\msep b(\term Q)} {[N]a.b<x>.P} {T_A} i }
\qquad	  
	  \bigstep  
	  { \eval[n] {R_A} N  {U_A} k \quad
	    \eval[m]       {S_A\msep a(\term N)}    {\{Q/x\}P} {T_A} i }
	  { \bigstep
	    { \eval[n+m]   {S_A}               {[N]a.\{Q/x\}P} {T_A} i }
	    { \eval[n+m+1] {S_A\msep b(\term Q)} {b<x>.[N]a.P} {T_A} i }
	  }
\]

	\item \emph{Case:}
\[
a)\quad
	\vc{
	\bigstep
	  { \eval[m]   {R_A}  M       {S_A} i \quad
	    \eval[n]   {S_A}  N       {T_A} k }
	  { \eval[m+n] {R_A} {M;i->N} {T_A} k }
	}
\qquad\qquad
b)\quad
	\vc{
	\bigstep
	  { \eval[m]   {R_A}  M       {S_A} j \quad
	    \eval[n]   {T_A}  N       {U_A} k }
	  { \eval[m+n] {R_A} {M;i->N} {S_A} j }
	}
\]
For each of the derivations $a)$ and $b)$ above there are five sub-cases, depending on the first reduction step: one for reduction inside $\term M$ or $\term N$, and four for a top-level reduction step, as follows.
\[
\begin{array}{l@{\qquad}r@{}l@{\qquad}l}
   1) & \term{M;i->N} &~\rw~\term{M';i->N'} & \text{where}~\term{M'}=\term{M}~\text{or}~\term{N'}=\term{N}
\\ 2) & \term{i;i->N} &~\rw~\term{N} & \text{i.e.}~\term{M}=\term{i},~\text{or}
\\    & \term{j;i->N} &~\rw~\term{j} & \text{i.e.}~\term{M}=\term{j}~\text{where}~i\neq j
\\ 3) & \term{(Q;i->P);i->N} &~\rw~\term{Q;i->(P;i->N)} & \text{i.e.}~\term{M}=\term{Q;i->P}
\\ 4) & \term{(a<x>.P);i->N} &~\rw~\term{a<x>.(P;i->N)} & \text{i.e.}~\term{M}=\term{a<x>.P}~\text{where}~x\notin\fv(\term N)
\\ 5) & \term{([Q]a.P);i->N} &~\rw~\term{[Q]a.(P;i->N)} & \text{i.e.}~\term{M}=\term{[Q]a.P}
\end{array}
\]

Case $1)$ follows by the outer inductive hypothesis on $\term{M'}$ or $\term{N'}$ and inner inductive hypothesis on the remaining reduction, similar to the first case for \emph{Push}. For case $2)$ evaluation is of the following forms respectively; $2a)$ follows by the outer inductive hypothesis on $\term{N}$, and $2b)$ is immediate from the evaluation $\evl[0]{R_A}j{R_A}j$.
\[
2a)\quad
	\vc{
	\bigstep
	  { \eval[0] {S_A}  i       {S_A} i \quad
	    \eval[n] {S_A}  N       {T_A} k }
	  { \eval[n] {R_A} {i;i->N} {T_A} k }
	}
\qquad\qquad
2b)\quad
	\vc{
	\bigstep
	  { \eval[0] {R_A}  j       {R_A} j \quad
	    \eval[n] {T_A}  N       {U_A} k }
	  { \eval[n] {R_A} {j;i->N} {R_A} j }
	}
\]
Cases $3)$ through $5)$ follow similarly to the third \emph{Push} case above, for passage reduction, by reconfiguring the evaluation derivation and observing that its size is preserved.
\end{itemize}
\end{proof}

The previous two lemmata then give typed strong normalization: the first shows that for a typed term a measured evaluation exists; the second that this measure bounds the number of beta-steps in the reduction of the term.

\begin{theorem}[Typed strong normalization]
If $\term{G |- M:t}$ then $\term M$ is strongly normalizing.
\end{theorem}

\begin{proof}
Let $\subs$ be a substitution map taking every $\term{x:s}$ in $\term{G}$ to a zero-term $\term{0_{\type s}}$. Then $\subs\in\EVS{G}$ by Lemma~\ref{lem:zero-eval} and $\term{\subs M}\in\EVL{t}$ by Lemma~\ref{lem:type-eval}. Let $\type{t}=\type{!!r=>!!sI}$. Lemma~\ref{lem:zero-eval} gives $0_{\type{!!r}}\in\EVL{!!r}$, and by the definition of $\EVL-$ there is an evaluation $\evl[m]{0_\type{!!r}}{\subs M}{S_A}i$. By Lemma~\ref{lem:eval-reduce} reduction on $\term M$ preserves the measure $m$ and reduces it in case of beta-steps. Since by Lemma~\ref{lem:aff} affine reduction is strongly normalizing, an infinite reduction from $\term M$ would have infinitely many beta-steps, a contradiction.
\end{proof}




\end{document}